\documentclass[11pt]{amsart}
\usepackage{graphicx}
\usepackage{amssymb,color}
\usepackage{epsfig}
\usepackage{epstopdf}
\usepackage{epsfig}
\usepackage{amsmath}
\usepackage{amsthm}

\allowdisplaybreaks

\def\cb{{\mathcal B}}

\def\cn{{\mathcal N}}

\def\bb{{\mathbb B}}

\def\bn{{\mathbb N}}

\def\br{{\mathbb R}}

\def\a{\alpha}
\def\b{\beta}
  
\def\d{\delta}  

\def\e{\eta}
\def\l{\lambda} 

\def\m{\mu}

\def\s{\sigma} 

\def\v{\varphi} \def\F{\Phi}

\def\w{\omega} \def\Om{\Omega}
\def\ab{\overline{a}} \def\bb{\overline{b}} \def\cb{\overline{c}}

\def\xd{x_{\downarrow}}
\def\Cbd{C_{b_{\downarrow}}}

\def\cbd{c_{b_{\downarrow}}}

\def\h{{\mathbf{h}}}

\newtheorem{thm}{Theorem}[section]
\newtheorem{lem}[thm]{Lemma}
\newtheorem{cor}[thm]{Corollary}

\newtheorem{defin}[thm]{Definition}

\newtheorem{rem}{Remark}[section]

\begin{document}
\title[Ground states and Phase Transition]{On Ground States and Phase Transition for $\lambda$-Model with the Competing Potts Interactions on Cayley Trees}

 \author{Farrukh Mukhamedov}
\address{Farrukh Mukhamedov\\
    ,
 Department of Mathematical Sciences, \& College of Science,\\
The United Arab Emirates University, Al Ain, Abu Dhabi,\\
15551, UAE} \email{{\tt far75m@yandex.ru}, {\tt
farrukh.m@uaeu.ac.ae}}

 \author{Chin Hee Pah}
\address{Chin Hee Pah\\
    ,
Department of Computational \& Theoretical Sciences\\
Faculty of Science, International Islamic University Malaysia\\
Kuantan, Pahang, Malaysia} \email{{\tt pahchinhee@gmail.com}}

\author{Hakim Jamil}

\address{Hakim Jamil\\
 Department of Computational \& Theoretical Sciences\\
Faculty of Science, International Islamic University Malaysia\\
Kuantan, Pahang, Malaysia} \email{{\tt m.hakimjamil@yahoo.com.my}}

\author{Muzaffar Rahmatullaev}

\address{Muzaffar Rahmatullaev\\
 Department of Mathematics, Namangan State University,\\
Namangan, Uzbekistan} \email{{\tt mrahmatullaev@rambler.ru }}

\maketitle

\begin{abstract}
In this paper, we consider the $\l$-model with nearest neighbor interactions and with competing Potts interactions on the Cayley tree of order-two.   We notice that if $\l$-function is taken as a Potts interaction function, then 
this model contains as a particular case of Potts model  with competing interactions on Cayley tree.  In this paper, we first describe all ground states of the model. 
We point out that the Potts model with considered interactions was
investigated only numerically, without rigorous (mathematical) proofs. One of
the main points of this paper is to propose a measure-theoretical approach for
the considered model in more general setting. Furthermore, we find certain conditions for the existence of Gibbs
measures corresponding to the model, which allowed to establish the existence
of the phase transition. 
\end{abstract}

\section{Introduction}

The main objective of statistical mechanics is to predict the relation between the observable macroscopic properties of the system given only the knowledge of the microscopic interactions between components. 
It can be explained by mathematical framework.It is known \cite{1} that the Gibbs measures are one of the central objects of equilibrium statistical mechanics. Also, one of the main problems of statistical physics is to describe all Gibbs measures corresponding to the given Hamiltonian \cite{2}. As is known, the phase diagram of Gibbs measures for a Hamiltonian is close to the
 phase diagram of isolated (stable) ground states of this Hamiltonian. At low temperatures, 
 a periodic ground state corresponds to a periodic Gibbs measure \cite{3,4}. The problem naturally lead to arises on description of periodic ground states. 

 A simplest model in statistical mechanics is the Ising model which has wide
theoretical interest and practical applications. There are several papers (see \cite{1,8} for review)
which are devoted to the description of this set for the Ising model on a Cayley tree.
However, a complete result about all Gibbs measures even for the Ising model is lacking.
Later on in \cite{V}  such an Ising model was considered with next-neatest neighbor
interactions on the Cayley tree for which its phase diagram was described.
On the other hand,  the q-state Potts model is one of the most studied models in
statistical mechanics due to its wide theoretical interest and practical
applications \cite{NS,2,DGM}. The Potts model \cite{9} was introduced as a
generalization of the Ising model to more than two components
and encompasses a number of problems in statistical physics (see,
e.g. \cite{10}). The model is structured richly enough to illustrate
almost every conceivable nuance of the subject. Furthermore, the Potts models became one of the important models in statistical mechanics.  These models describe a special class of statistical mechanics systems, which are quite simply defined.

The Potts model with competing interactions on the Cayley tree  is more complex and has rich structure of ground states \cite{5,6,7} (see also \cite{8}). Nevertheless, their structure is sufficiently rich to describe almost every conceivable nuance of an object of investigation. In \cite{GMP}  a phase diagram of the
three-state Potts model with competing nearest neighbor and
next nearest neighbor interactions on a Cayley tree has been obtained (numerically). On the other hand, the structure of the Gibbs measures of the Potts models was investigated in \cite{11,12,13}. 
It is natural to consider more complicated models than the Potts one, so called $\l$-model \cite{14,15}. In \cite{16,17}  we have investigated the set of ground states for $\l$-model (with nearest neighbor interactions) on Cayley tree.  Furthermore,
the phase transition has been also established for the mentioned model \cite{18}.  

To the best knowledge of the authors, q-state Potts model with
competing interactions on the Cayley tree is not well studied from the measure-theoretical point of view. Some particular cases have been carried out when the competing interactions are located in the same level of the tree \cite{5,6,7}.   
Therefore, one of the main aims of the present paper is to develop
a measure-theoretic approach (i.e. Gibbs measure formalism) to rigorously establish
the phase transition for the $\l$-model with competing Potts interactions on the Cayley tree. We notice that until now, many researchers have investigated Gibbs measures corresponding to the Ising types of 
models \cite{MAK}. The aim of this paper is to
propose rigorously the investigation of Gibbs measures for the$\l$-model with competing Potts interactions which include as a particular case of Potts model with competing interactions. 

The paper is organized as follows.  In section 2, we provide necessary notations and
define the $\l$-model with competing Potts interactions on Cayley tree of order two. In section 3, we describe ground states of the considered model. 
In section 4, using a rigorous
measure-theoretical approach, we find certain conditions for the existence of Gibbs measures
corresponding to the model on the Cayley tree. To describe
the Gibbs measure, we obtain a system of functional equations (which is extremely
difficult to solve). Nevertheless, we are able to succeed in obtaining explicit
solutions by making reasonable assumptions, for the existence of translational invariant
Gibbs measures which allows us to establish the existence of the phase transition. We point out that when the competing Potts interaction vanishes, then the model reduced to the $\l$-model which was investigated in \cite{11,18}.

\section{Preliminaries}
Let $\tau^k = (V,L)$ be a Cayley tree of order
$k$, i.e, an infinite tree such that exactly $k+1$ edges are incident to each vertex. Here $V$ is
the set of vertices and $L$ is the set of edges of $\tau^k$.

Let $G_k$ denote the free product of $k+1$ cyclic groups $\{e, a_i\}$ of order 2 with generators $a_1, a_2,\dots, a_{k+1}$,
i.e., let $a^2_i=e$ (see \cite{13}).

There exists a one-to-one correspondence between the set $V$ of vertices of the Cayley tree of order $k$ and the group $G_k$ \cite{8}.

For the sake of completeness, let us establish this correspondence (see \cite{8} for details). We choose an arbitrary vertex $x_0 \in V  $and associate it with the identity element $e$ of the group $G_k$. Since we may assume that the graph under consideration is planar, we associate each neighbor of $x_0$ (i.e., $e$) with a single generator $a_i, i=1, 2,\dots, k + 1$, where
the order corresponds to the positive direction, see Figure \ref{cayley}.
\begin{figure}[h!]
    \begin{center}
        \includegraphics[width=12cm]{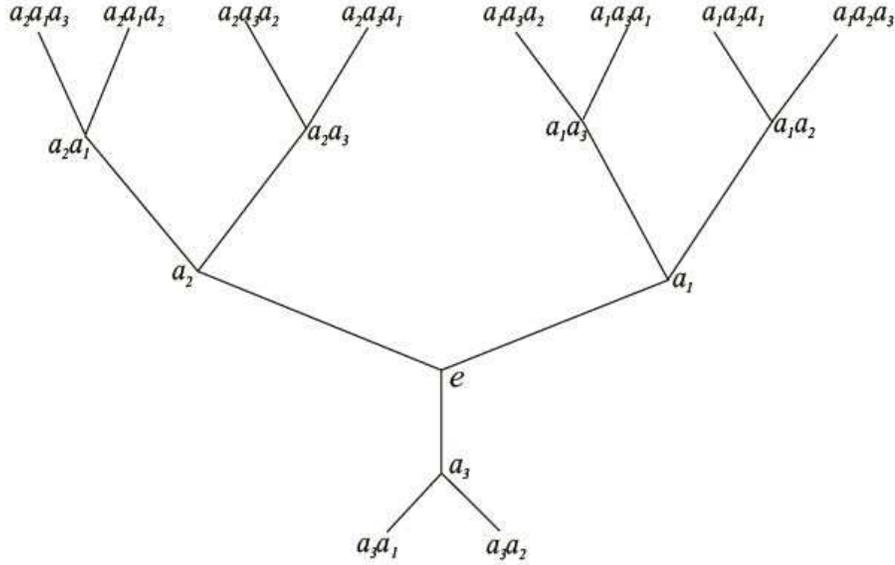}
    \end{center}
    \caption{The Cayley tree $\tau^2$ and elements of the group representation of vertices} \label{cayley}
\end{figure}

For every neighbor of $a_i$, we introduce words of the form $a_ia_j$ . Since one of the neighbors of $a_i$ is $e$, we put $a_ia_i = e$. The remaining neighbors of $a_i$ are labeled according to the above order. For every neighbor of $a_ia_j$ , we introduce words of length 3 in a similar way. Since one of the neighbors of $a_ia_j$ is $a_i$, we put $a_ia_ja_j = a_i$. The remaining neighbors of $a_ia_j$ are labeled by words of the form $a_ia_ja_l$, where $i, j, l = 1, 2, . . . , k + 1$, according to the above procedure. This agrees with the previous stage
because $a_ia_ja_j = a_ia^2_j= a_i$. Continuing this process, we obtain a one-to-one correspondence between the vertex set of the Cayley tree $\tau^k$ and the group $G_k$.

The representation constructed above is said to be $right$ because,
for all adjacent vertices $x$ and $y$ and the corresponding elements
$g,h \in G_k,$ we have either $g = ha_i$ or $h=ga_j$ for suitable
$i$ and $j$. The definition of the $left$ representation is similar.

For the group $G_k$ (or the corresponding Cayley tree), we consider the left (right) shifts. For $g \in G_k$,
we put

\begin{equation*}
    T_g(h)=gh \ (T_g(h)=hg) \ for \ all \ h \in G^k
\end{equation*}
The group of all left (right) shifts on $G_k$ is isomorphic to the group $G_k$.

Each transformation $S$ on the group $G_k$ induces a transformation $S$ on the vertex set $V$ of the Cayley
tree  $\tau^k$. In the sequel, we identify $V$ with $G_k$.

\begin{thm}
    The group of left (right) shifts on the right (left) representation of the Cayley tree is the group of translations.
\end{thm}

By the group of translations we mean the automorphism group of the
Cayley tree regarded as a graph. Recall that a mapping $\psi$ on the vertex set of a graph G
is called an automorphism of G if $\psi$ preserves the adjacency
relation, i.e., the images $\psi(u)$  and $\psi(v)$ of vertices $u$
and $v$ are adjacent if and only if $u$ and $v$ are adjacent.

For an arbitrary vertex $x_0 \in V $, we put

\begin{equation*}
    W_{n}=\left\{ x\in V\mid d(x^0,x)=n\right\}, \ \
    V_n=\bigcup\limits_{m=0}^{n}W_m , \ \
    L_{n}=\left\{
    l=<x,y>\in L\mid x,y\in V_{n}\right\}.
\end{equation*}
where $d(x, y)$ is the distance between $x$ and $y$ in the Cayley tree, i.e., the number of edges of the path
between $x$ and $y$.

For each $x \in G_k$, let $S(x)$ denote the set of immediate successor of $x$, i.e., if $x \in W_n$ then

$$
S(x)=\left\{ y\in W_{n+1}:d(x,y)=1\right\}.
$$

For each $x \in G_k$, let $S_1(x)$ denote the set of all neighbors of $x$, i.e., $S_1(x)=\{y \in G_k:<x,y> \in \ L \}$.
The set $S_1(x) \setminus S(x)$ is a singleton. Let $\xd$ denote the (unique) element of this set.

Assume that spin takes its values in the set $\F=\{1, 2,\dots, q\}.$
By a configuration $\s$ on $V$ we mean a function taking $\sigma:x
\in V\to \s(x) \in \F.$ The set of all configurations coincides with
the set $\Omega = \F^V$.

Consider the quotient group $G_k/G^*_ k = \{H_1,\dots,H_r\}$,
where $G^*_k$ is a normal subgroup of index $r$ with $r \geq 1$.

\begin{defin}
    A configuration $\s(x)$ is said to be $G^*_k$-periodic if $\s(x)
    =\s_i$ for all $x \in G_k$ with $x \in H_i$. A $G_k$-periodic
    configuration is said to be translation invariant.
\end{defin}
By \textit{period} of a periodic configuration we mean the index of
the corresponding normal subgroup.


\begin{defin}
    The vertices $x$ and $y$ are called next-nearest-neighbor which is denoted by ${>x,y<}$, if there exists a vertex $z \in V$ such that $x,z$ and $y,z$ are nearest-neihbors.
\end{defin}

Let spin variables $\sigma(x), x\in V,$ take values $\{1, 2, 3\}$.
The  $\lambda$-Model with competing Potts Interactions is
defined by the following Hamiltonian:
\begin{equation}\label{hm}
H(\sigma)=\sum_{<x,y>}\lambda(\sigma(x),\sigma(y))+J\sum_{{>x,y<}}\delta_{\sigma(x),\sigma(y)},
\end{equation}
where $J\in {R}$ and  $\delta$ is the Kronecker  symbol and
\begin{equation}\label{cond}
\l(i,j)=\left\{
\begin{array}{lll}
\ab,& \ \textrm{if} \ \ &|i-j|=2,\\
\bb,&\ \textrm{if} \ \ &|i-j|=1,\\
\cb,&\ \textrm{if} \ \ &i=j,
\end{array}\right.
\end{equation}
where $\ab,\bb,\cb\in \mathbb{R}$  some given numbers. 

We notice if $\l(i,j)=J_0\d_{i,j}$, (where $J_0$ is some constant. In this setting, we have $\ab=\bb=0$, $\cb=J_0$) then the model reduces to the Potts model with competing interactions which was numerically investigated in \cite{GMP}.
Moreover, if one takes $\l(i,j)=J_0|i-j|$, then the model reduces to Solid-on-Solid (SOS) model with competing Potts interactions. Some analogue of this model has been recently studied in \cite{Ras}.

\section{Ground States}

In this section, we are going to describe  ground state of the  $\lambda$-Model
with competing Potts interactions on a Cayley tree of order two. 

For a pair
of configurations $\s$ and $\v$ which coincide almost everywhere,
i.e., everywhere except finitely many points, we consider a
relative Hamiltonian $H(\s,\v)$ determining the energy differences
of the two configurations $\s$ and $\v$:

\begin{eqnarray}\label{eq12}
H(\s,\v)=\sum_{\substack{<x,y>\\ x,y \in V}}(\l(\s(x),\s(y))-\l(\v(x),\v(y)))+J\sum_{\substack{>x,y<\\ x,y \in V}}(\d_{\s(x),\s(y)}-\d_{\v(x),\v(y)})
\end{eqnarray}

Let $M$ be the set of unit balls with vertices in $V$, i.e. $M=\{x\in S_1(x): \forall x\in V$\}. The restriction of a configuration $\s$ to the ball $b\in M$ is called 
{\it bounded configuration} $\s_b .$

We define the energy of a configuration $\s_b$ on $b$ as follows:
\[U(\s_b)=\dfrac{1}{2}\sum_{\substack{<x,y>\\ x,y \in V}}\l(\s(x),\s(y))+J\sum_{\substack{>x,y<\\ x,y \in V}}\d_{\s(x),\s(y)}.\]
From \eqref{eq12}, we get the following lemma.

We shall say that two bounded configurations $\s_b$ and $\s'_{b'}$
belong to the same class if $U(\s_b)=U(\s'_{b'})$ and they are denoted by 
$\s'_{b'}\sim \s_b. $

\begin{lem}
    Relative Hamiltonian \eqref{eq12} has the form

    \[H(\s,\v)=\sum_{b \in M}(U(\s_b)-U(\v_b)).\]
\end{lem}

For any configuration $\s_b$, we have
\[U(\s_b)\in\{U_1,U_2,U_3,U_4,U_5,U_6,U_7,U_8,U_9,U_{10},U_{11},U_{12}\},\]

where
\begin{equation}\label{U_n}
\begin{array}{lllll}
U_1=3\cb/2+3J,&U_2=(2\cb+\bb)/2+J,\\
U_3=(2\cb+\ab)/2+J,&U_4=(2\bb+\ab)/2+J,\\
U_5=(2\ab+\cb)/2+J,&U_6=3\ab/2+3J,\\
U_7=3\bb/2+J,&U_8=(2\bb+\cb)/2,\\
U_9=(2\ab+\bb)/2+J,&U_{10}=(\ab+\bb+\cb)/2\\
U_{11}=3\bb/2+3J,&U_{12}=(2\bb+\cb)/2+J.
\end{array}
\end{equation}

\begin{defin}\label{gs}
    A configuration $\v$ is called a ground state of the relative Hamiltonian H if
    \begin{eqnarray}\label{eq14}
      U(\v_b)=\min\{U_1,U_2,U_3,U_4,U_5,U_6,U_7,U_8,U_9,U_{10},U_{11},U_{12}\},
    \end{eqnarray}
    for any $b \in M.$

If a ground state is a periodic configuration then we call it a periodic ground state.
\end{defin}

By denoting
\begin{eqnarray}\label{eq15}
A_m=\{(\ab,\bb,\cb,J)\in \mathbb{R}^4|\  U_m=\min_{1 \leq k \leq 12}\{U_k\}\},
\end{eqnarray}
from \eqref{U_n}, we easily obtain
\begin{align*}
A_1=&\left\{(\ab,\bb,\cb,J)\in \mathbb{R}^4| \ \ab\geq\cb, \bb \geq \cb,J \leq\min \left\{\frac{\ab-\cb}{4}, \frac{\bb-\cb}{4},\frac{\ab+\bb-2\cb}{6}\right\}\right\}\\
A_2=&\left\{(\ab,\bb,\cb,J)\in \mathbb{R}^4| \ \ab\geq\bb\geq\cb, \frac{\bb-\cb}{4} \leq J \leq\frac{\bb-\cb}{2}\right\},\\
A_3=&\left\{(\ab,\bb,\cb,J)\in \mathbb{R}^4| \ \bb\geq \ab\geq \cb, \frac{\ab-\cb}{4} \leq J \leq\frac{\bb-\cb}{2} \right\},\\
A_4=&\left\{(\ab,\bb,\cb,J)\in \mathbb{R}^4| \ \ab=\bb, \cb \geq \ab, 0\leq J\leq \frac{\cb-\ab}{2}\right\},\\
A_5=&\left\{(\ab,\bb,\cb,J)\in \mathbb{R}^4| \ \bb\geq \cb\geq \ab, \frac{\cb-\ab}{4} \leq J \leq\frac{\bb-\ab}{2}\right\},\\
A_6=&\left\{(\ab,\bb,\cb,J)\in \mathbb{R}^4| \ \cb\geq\ab, \bb \geq \ab,J \leq\min \left\{\frac{\cb-\ab}{4}, \frac{\bb-\ab}{4},\frac{\bb+\cb-2\ab}{6}\right\}\right\},\\
A_7=&\left\{(\ab,\bb,\cb,J)\in \mathbb{R}^4| \ \ab\geq \bb,\cb\geq \bb, 0\leq J\leq \frac{\cb-\bb}{2}\right\},\\
A_8=&\left\{(\ab,\bb,\cb,J)\in \mathbb{R}^4| \ \ab\geq \bb, 0\leq J, |\bb-\cb|\geq 2J, \cb-\ab\leq 2J\right\},\\
A_9=&\left\{(\ab,\bb,\cb,J)\in \mathbb{R}^4| \ \cb\geq \bb\geq \ab, \frac{\bb-\ab}{4} \leq J \leq\frac{\cb-\ab}{2}\right\},\\
A_{10}=&\big\{(\ab,\bb,\cb,J)\in \mathbb{R}^4| \ 0\leq \bb-\ab\leq 2J, |\ab-\cb|\leq 2J, |\bb-\cb|\leq 2J\big\},\\
A_{11}=&\big\{(\ab,\bb,\cb,J)\in \mathbb{R}^4| \ \cb\geq \bb, \ab\geq\bb, J\leq 0\big\},\\
A_{12}=&\big\{(\ab,\bb,\cb,J)\in \mathbb{R}^4| \ \cb = \bb, \ab\geq\bb, J= 0\big\}.
\end{align*}

To construct ground states, let us denote for a given ball $b$ a configuration on it as follows: $x_b,y_b,c_b, \cbd \in (\ab,\bb,\cb)$:

        \begin{figure}[h!]
            \begin{center}
                \includegraphics[width=3.5cm]{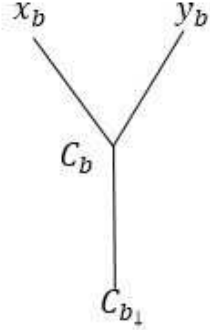}
            \end{center}
            \caption{A ball,b} \label{ball}
        \end{figure}



Let us introduce some notations. We put
$$C_i=\{\s_b \in \Om_b: \
U(\s_b)=U_i \}, \ \ \ i=\overline{1,10}$$ and $B^{(i)}=|\{x \in
S_1(k):\ \v_b(x)=i\}|$ for $i=\overline{1,3}$.

Let $A\subset \{1,2,...,k+1\}$, $H_A=\{x\in G_k: \sum_{j\in
    A}w_j(x)-$even$\},$ where $w_j(x)$-is the number of letters $a_j$
in the word $x.$

It is obvious, that $H_A$ is a normal subgroup of index two. Let $G_k/H_A=\{H_A,G_k\setminus H_A\}$ be the quotient
group. We set $H_0=H_A, H_1=G_k\setminus H_A$.


\begin{thm}
    Let $(\ab,\bb,\cb) \in A_1$, then there are only three ground states which are translation-invariant.
\end{thm}

\begin{proof}
    Let $(\ab,\bb,\cb) \in A_1$, then one can see that for this triple, the minimal value is $\frac{3\cb}{2}+3J$, which is achieved by the configuration on $b$ (see  Figure \ref{case1})
    \begin{figure}[h!]
        \begin{center}
            \includegraphics[width=7.5cm]{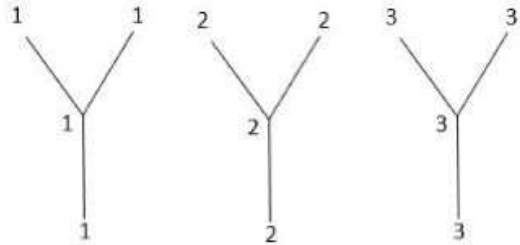}
        \end{center}
        \caption{Configurations for $A_1$.} \label{case1}
    \end{figure}

    In this case, we have three configurations
    \begin{eqnarray*}
        \s^{(k)}(x)=k, \ \ \ \forall x\in V, &k=\{1,2,3\}.
    \end{eqnarray*}
 which are translation-invariant ground states.
\end{proof}


Let $H_{a_1}=\{{x \in G_2:w_1(x) \ \ \mbox{is even}}\}$ where $w_1(x)$ is the number of letter $a_1$ in word $x \in G_2$. Note that the $H_{a_1}$ is a normal subgroup of group $G_2$ (see \cite{8} ).

\begin{thm}\label{th2}
    Let $(\ab,\bb,\cb) \in A_2$, then the following statements hold:
    \begin{itemize}
        \item[(i)] there is uncountable number of ground states;
        \item[(ii)] there exist four $H_{\{a_1\}}-$periodic ground states.
    \end{itemize}
\end{thm}

\begin{proof}
 Let $(\ab,\bb,\cb) \in A_2$, then the minimal value of $U(\s_b)$ is $(2\cb+\bb)/2+J$,
    which is achieved by the configurations on $b$ given in Figures \ref{case2a} and \ref{case2b}.

        \begin{figure}[h!]
            \centering
            \begin{minipage}[b]{0.4\textwidth}
                \includegraphics[width=\textwidth]{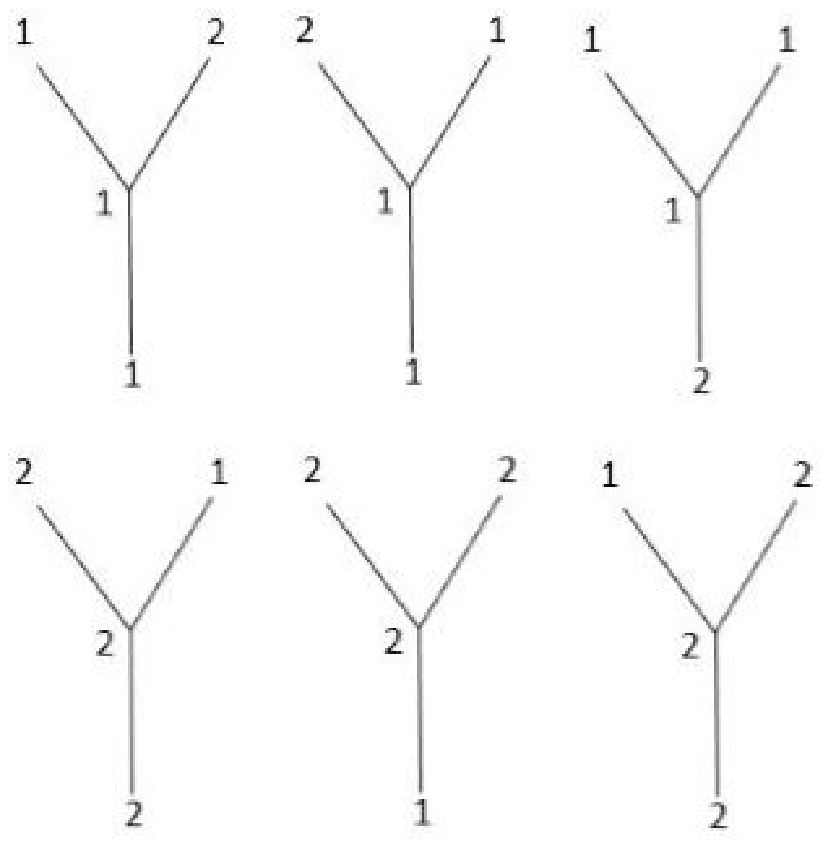}
                \caption{$\Gamma_1$-Configuration for $A_2$}\label{case2a}
            \end{minipage}
            \hfill
            \begin{minipage}[b]{0.4\textwidth}
                \includegraphics[width=\textwidth]{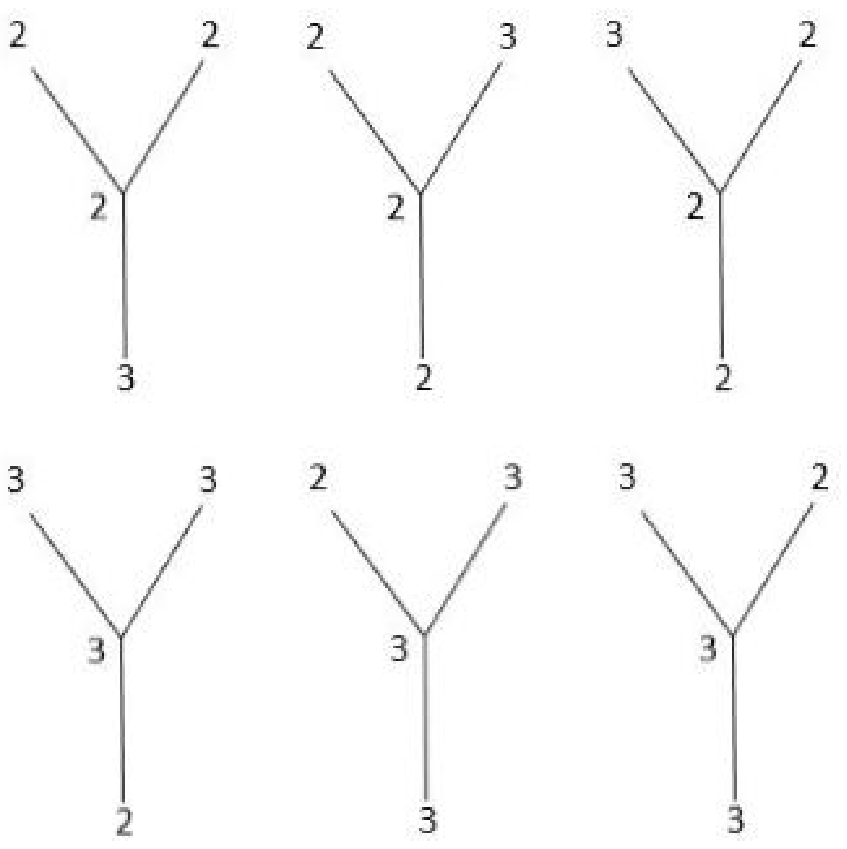}
                \caption{$\Gamma_2$-Configuration for $A_2$}\label{case2b}
            \end{minipage}
        \end{figure}

    \begin{itemize}
        \item[(i)] Let us construct ground states by means of configurations given by Figures \ref{case2a} and \ref{case2b}:

        \begin{figure}[h!]
            \begin{center}
                \includegraphics[width=7.5cm]{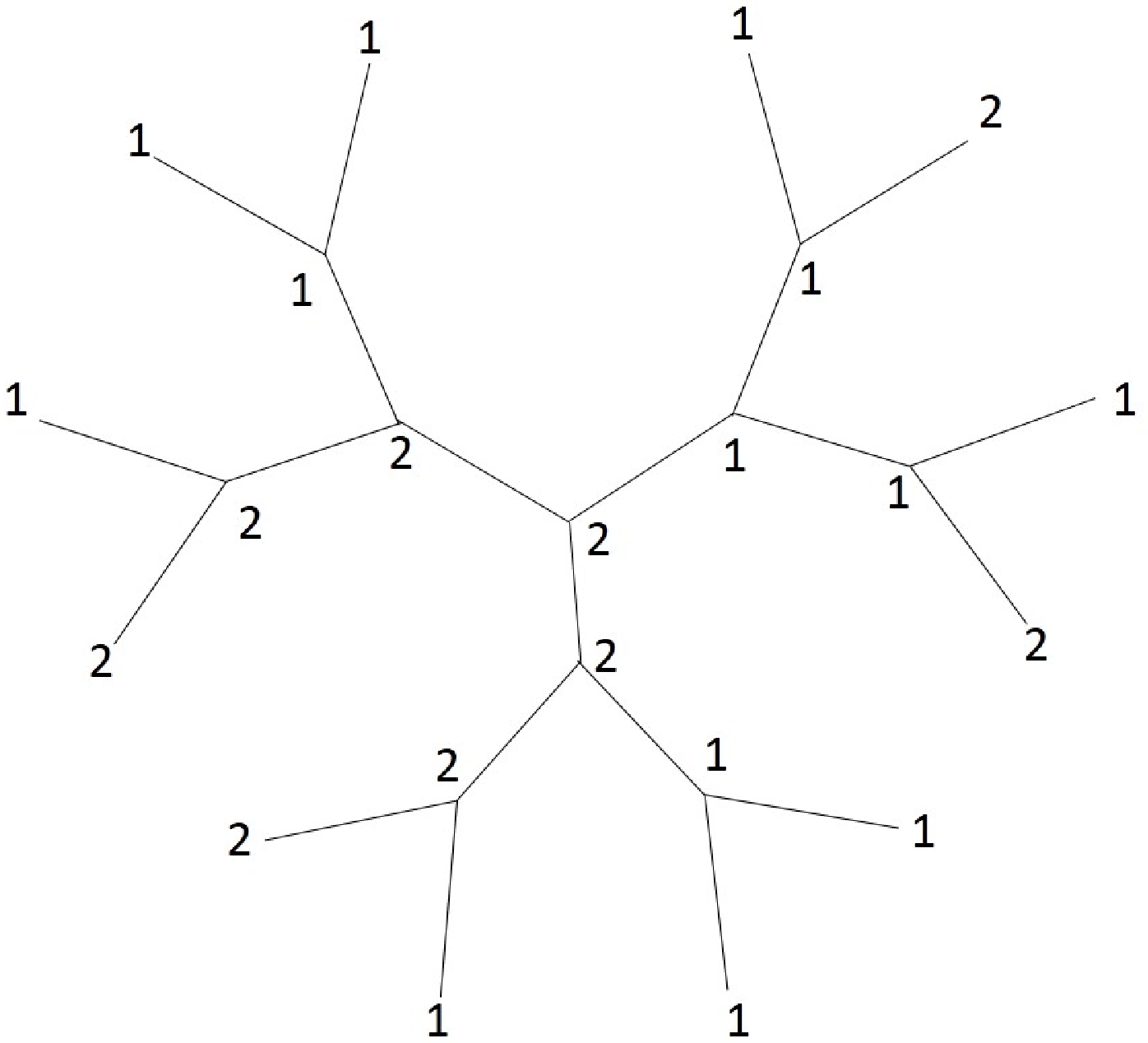}
            \end{center}
            \caption{Example for Cayley tree by \ref{case2a}} \label{example2A}
        \end{figure}

        We choose an initial ball $b$, and let $\s_b$ be a configuration on $b$. Let us consider several cases with respect to $\Cbd,C_b$ and $\s_b$.

        \begin{itemize}
            \item[Case(1)] Let $\s(\Cbd)=2$ and $\s(C_b)=2$, then we can construct different combinations by choosing $\s(x_b)$ and $\s(y_b)$ as follows:

            $(i_1)$ $\s(x_b)=1$ and $\s(y_b)=2$, $(i_2)$ $\s(x_b)=2$ and $\s(y_b)=3$.

            In case (1), we need to plug configuration from $\Gamma_1$ on ball $b_1$ and for the ball $b_2$
             we can plug configurations  from  $\Gamma_1$ and $\Gamma_2$ by the following rule:

            \begin{itemize}

                \item[(a)]$\s(C_{{b_1}{\downarrow}})=2$, $\s(C_{b_1})=1$, for which we only have possibility $\s(x_{b_1})=1$, $\s(y_{b_2})=1$.

                \item[(b)] $\s(C_{{b_2}{\downarrow}})=2$, $\s(C_{b_1})=2$, for which we have again to possibilities $(i_1)$, $(i_2)$ as above.
                In this case, to plug the configuration with $\s(C_{{b}{\downarrow}})=2$, $\s(C_{b})=2$. When $\s(x_b)=2$,
                we are again in the same situation what we are considering. If $\s(x_b)=1$, the further plug configuration from $\Gamma_1$, we have only one possibility.
                Hence, this is reduced to Case (1).
            \end{itemize}

        \item[Case(2)] Let $\s(\Cbd)=1$ and $\s(C_b)=2$. In this case, we only have one possibility, $\s(x_b)=2$,$\s(y_b)=2$.
        It is easy to see that in this case, we immediately reduce to the case which was considered above (see case (b)).
        \item[Case(3)] Let $\s(\Cbd)=2$ and $\s(C_b)=3$. In this case, we only have one possibility, $\s(x_b)=3$, $\s(y_b)=3$.
        Let $\s(\Cbd)=3$ and $\s(C_b)=3$. In this case, we only have one possibility, $\s(x_b)=3$, $\s(y_b)=2$.
        Let $\s(\Cbd)=3$ and $\s(C_b)=2$. In this case, we only have one possibility, $\s(x_b)=2$, $\s(y_b)=2$.
        It is easy to see that in this case, we immediately reduce to the case which was considered above (see (b)). Then there uncountable number of ground states.

        \end{itemize}
        We can construct ground states using only configurations given by $\Gamma_1$ (see figure \ref {example2A}).

        \item[(ii)]  We consider the quotient group $G_2/H_{\{a_1\}}=\{H_0, H_1\},$ where
        $$
        H_0=H_{\{a_1\}}, H_1=\{{x \in G_2:w_1(x) \ \ \mbox{is odd}}\}.
        $$

        Let
        \begin{equation}\label{1}
        \v(x)=\left\{
        \begin{array}{ll}
        i, \ \mbox{if} \ x \in H_0,\\
        j, \ \mbox{if} \ x \in H_1,
        \end{array}\right.
        \end{equation}
        be a $H_{a_1}-$periodic configuration (see Figure \ref{percase2}), where $|i-j|=1$.
        We are going to prove that $\v$ is a  ground state. Let $b\in M$ be an arbitrary unit ball and $C_b\in H_0$, then it is easy to see that 
        $|\{C_{{b}{\downarrow}}, x_b,y_b\}\cap H_0|=2$ and $|\{C_{{b}{\downarrow}}, x_b,y_b\}\cap H_1|=1$. In this case, there are the following possibilities:

        1) $\v(C_b)=i, \v(C_{{b}{\downarrow}})=i, \v(x_b)=i, \v(y_b)=j$;\\
        2) $\v(C_b)=i, \v(C_{{b}{\downarrow}})=i, \v(x_b)=j, \v(y_b)=i$;\\
        3) $\v(C_b)=i, \v(C_{{b}{\downarrow}})=j, \v(x_b)=i, \v(y_b)=i$;\\
        In all cases $U(\v_b(x))=(2\ab+\bb)/2+J$.

      If $C_b\in H_1$, then it is easy to see that $|\{C_{{b}{\downarrow}}, x_b,y_b\}\cap H_0|=1$ and $|\{C_{{b}{\downarrow}}, x_b,y_b\}\cap H_1|=2$.
      Again, in this setting, we have the following possibilities:

        1) $\v(C_b)=j, \v(C_{{b}{\downarrow}})=i, \v(x_b)=j, \v(y_b)=j$;\\
        2) $\v(C_b)=j, \v(C_{{b}{\downarrow}})=j, \v(x_b)=i, \v(y_b)=j$;\\
        3) $\v(C_b)=j, \v(C_{{b}{\downarrow}})=j, \v(x_b)=j, \v(y_b)=i$;\\
        As before, in all cases, one has $U(\v_b(x))=(2\cb+\bb)/2+J$, i.e. $\v_b \in C_2, \forall b\in M$.  Hence, the periodic configuration $\v$ is a ground state.

        \begin{figure}[h!]
            \begin{center}
                \includegraphics[width=10.5cm]{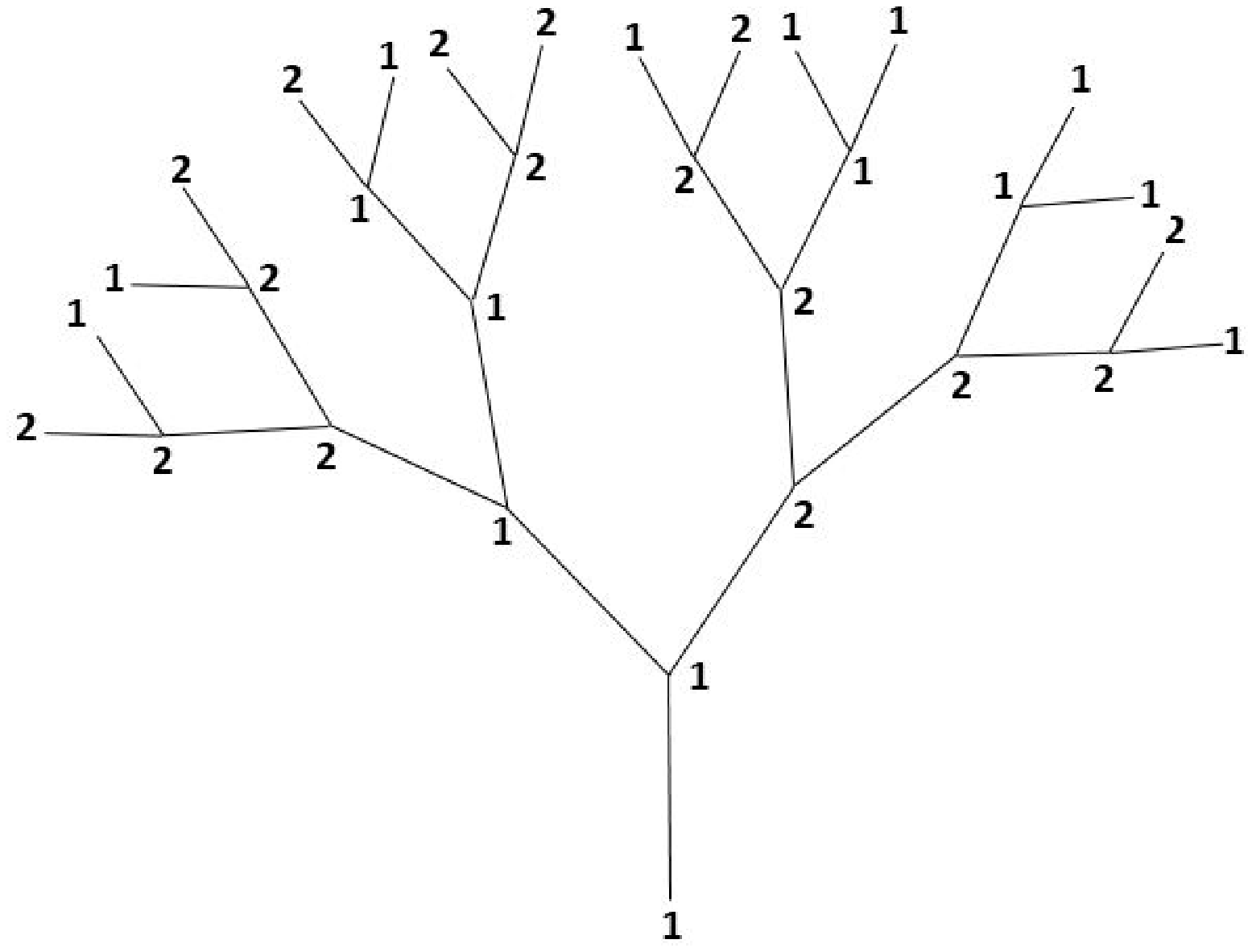}
            \end{center}
            \caption{Reduced Cayley Tree for $A_2$} \label{percase2}
        \end{figure}

        \end{itemize}
\end{proof}

\begin{thm}\label{th1}
    Let $(\ab,\bb,\cb, J) \in A_3$, then there exist only two $H_{\{a_1\}}$-periodic ground states.
\end{thm}

\begin{proof}

Let $(\ab,\bb,\cb) \in A_3$, then one can see that for this triple, the minimal value is $(2\cb+\ab)/2+J$, which is achieved by the configurations on b given by Figure \ref{case3}.

\begin{figure}[h!]
    \begin{center}
        \includegraphics[width=7cm]{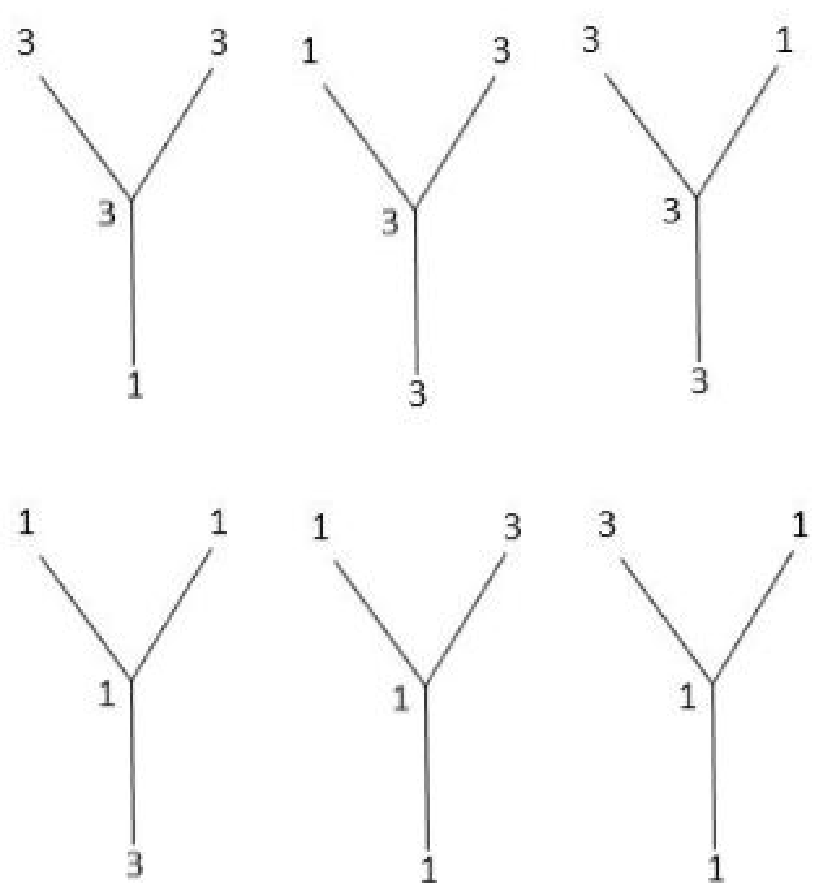}
    \end{center}
    \caption{Configurations for $A_3$} \label{case3}
\end{figure}

   In this case we can construct only two configuration $\s$, which $\s_b\in C_3, \forall b\in M.$ We choose an initial ball $b$ and $\s_b$ from Figure \ref{case3}.
   Let us $\s(C_{{b}{\downarrow}})=1$ and $\s(C_b)=3$, then we have only one case $\s(x_b)=3, \s(y_b)=3$. If $\s(C_{{b}{\downarrow}})=3$ and $\s(C_b)=3$,
   then we have the following cases $\s(x_b)=3, \s(y_b)=1$ or $\s(x_b)=1, \s(y_b)=3$.  Now, we notice that if one interchanges the trees issues from the vertices $x_b$ and $y_b$, respectively, then
   the configuration does not change.    
   ]Therefore,  in both cases we have the same configuration.
    If $\s(C_{{b}{\downarrow}})=3$ and $\s(C_b)=1$, then we have only one case $\s(x_b)=1, \s(y_b)=1$. If $\s(C_{{b}{\downarrow}})=1$ and $\s(C_b)=1$, then we
    have the following cases $\s(x_b)=3, \s(y_b)=1$ or $\s(x_b)=1, \s(y_b)=3$. Again using above notice, in both cases one gets the same configuration.

   It is easy to see that this configurations are $H_{\{a_1\}}-$periodic and have the form

   \begin{equation}\label{2}
    \v_{i,j}(x)=\left\{
    \begin{array}{ll}
    i&,x \in H_0,\\
    j&,x \in H_1,
    \end{array}\right.
    \end{equation}
   where $|i-j|=2.$  Using the argument of the proof of Theorem \ref{th2} we can prove that configurations $\v_{i,j}$ are ground states.
    Note that a number of configurations $\v_{i,j}$, (with $ |i-j|=2, i,j\in \Phi$) is two. For example, the configuration $\v_{1,3}$ is presented in Figure \ref{percase3} on reduced Cayley tree. 
    \begin{figure}[h!]
        \begin{center}
            \includegraphics[width=10.5cm]{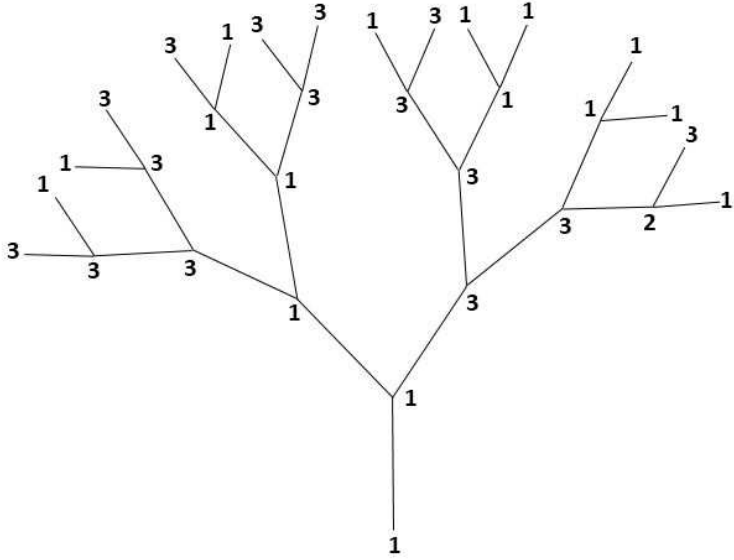}
        \end{center}
        \caption{Reduced Cayley Tree for $A_3$} \label{percase3}
    \end{figure}
\end{proof}

\begin{thm}\label{th3}
    Let $(\ab,\bb,\cb, J) \in A_4$, then there is not ground states.
\end{thm}

\begin{proof}
Let $(\ab,\bb,\cb, J) \in A_4$, then one can see that for this triple, the minimal value is $(2\bb+\ab)/2+J$.
Let $\v(x)$ configuration be a ground state, which is for any $b\in M$, $\v_b(x)\in C_4$. Then it must be $\v(C_b)=1$ or $\v(C_b)=3$, because $U(\v_b)=(2\bb+\ab)/2+J$. From $\v_b(x)\in C_4$ we have, that one of the following variables $|\v(C_b)-\v(C_{{b}{\downarrow}})|, |\v(C_b)-\v(x_b)|, |\v(C_b)-\v(y_b)|$ must be equal to 2 and another two are equal to 1. Then some two of $\v(C_{{b}{\downarrow}}), \v(x_b), \v(y_b)$ equal to 2. But we do not have $b'\in M$ with $\v(C_{b'})=2$ and $\v_{b'}\in C_4.$ 
Consequently, there is not any ground state.
\end{proof}

Let  $G_2^{\{4\}}=\{x \in G_2:w_x(a_1)-even,\ w_x(a_2)-even\}$.
Note that $G_2^{\{4\}}$ is a normal subgroup of index four.
Consider the following quotient group $G_2/G_2^{\{4\}}=\{H_0, H_1,
H_2, H_4\},$ where
$$
\begin{array}{llll}
           H_0=G_2^{\{4\}}\\
            H_1=\{x \in G_2:w_x(a_1)-even,\ w_x(a_2)-odd\},\\
            H_2=\{x \in G_2:w_x(a_1)-odd,\ w_x(a_2)-even\},\\
            H_3=\{x \in G_2:w_x(a_1)-odd,\ w_x(a_2)-odd\}.
        \end{array}$$

\begin{thm}\label{th}
    Let $(\ab,\bb,\cb, J) \in A_5$, then there exist only two $G_2^{\{4\}}$-periodic ground states for index 4.

\end{thm}

\begin{proof}
 Let $(\ab,\bb,\cb,J) \in A_5$, then one can see that for this triple, the minimal value is $(2\ab+\cb)/2+J$.

In this case, we can construct only two configuration $\s$, for which
$\s_b\in C_5, \forall b\in M.$ Let $b$ be any initial ball from
$M$. Let $\s(C_{{b}{\downarrow}})=1$ and $\s(C_b)=1$, then we
have only one case with $\s(x_b)=3, \s(y_b)=3$. If
$\s(C_{{b}{\downarrow}})=1$ and $\s(C_b)=3$, then one finds the
following cases: $\s(x_b)=3, \s(y_b)=1$ or $\s(x_b)=1, \s(y_b)=3$.
Here, we are again in the same situation as in the proof of Theorem \ref{th1}. Hence, in both cases we have the same
configuration. If $\s(C_{{b}{\downarrow}})=1$ and $\s(C_b)=1$,
then we have only one case $\s(x_b)=3, \s(y_b)=3$. If
$\s(C_{{b}{\downarrow}})=3$ and $\s(C_b)=1$, then one has the
following cases: $\s(x_b)=3, \s(y_b)=1$ or $\s(x_b)=1, \s(y_b)=3$.
Here, again using above argument, we obtain the same
configuration.

   It is easy to see that these configurations are $G_2^{\{4\}}-$periodic and have the form

        \begin{equation}\label{3}
        \v^{(5)}_{i,j}(x)=\left\{
        \begin{array}{llll}
        i&,x \in H_0\cup H_3,\\
        j&,x \in H_1\cup H_2.
        \end{array}\right.
        \end{equation}
where $|i-j|=2.$

Indeed, let $b$ be any initial ball from
$M$ and $C_b\in H_0$ then one element of the set
$\{C_{{b}{\downarrow}}, x_b, y_b\}$ belongs to class $H_0$, one
element belongs to the class $H_1$ and another one element belongs
to the class $H_2$, i.e. $(\v^{(5)}_{i,j})_b\in C_5.$ By the similar
way, for $b\in M$, which $C_b\in H_l, l=1,2,3$ we can prove that
$(\v^{(5)}_{i,j})_b\in C_5.$
    Note that a number of the configurations $\v^{(5)}_{i,j}, |i-j|=2, i,j\in \Phi$ is two. 
    
       We reduce our tree as below:
  We have all configuration $\v (x) \in A_5$, correspondingly there exist $H_0$-periodic ground states.
    \end{proof}

Let $G_2^{(2)}=\{x\in G_2 : |x|\ \ \mbox {is even}\}.$ Notice that
$G_2^{(2)}$ is a normal subgroup of index two of $G_2$ (see
\cite{14}).

\begin{thm}\label{th4}
    Let $(\ab,\bb,\cb, J) \in A_6$, then there are only two $G_2^{(2)}-$periodic ground states.
\end{thm}

\begin{proof}

    Let $(\ab,\bb,\cb,J) \in A_6$, then one can see that for this triple, the minimal value is $(3\ab)/2+3J$, which is achieved by the configurations on $b$ given in Figure \ref{case6}.
\begin{figure}[h!]
    \begin{center}
        \includegraphics[width=5.5cm]{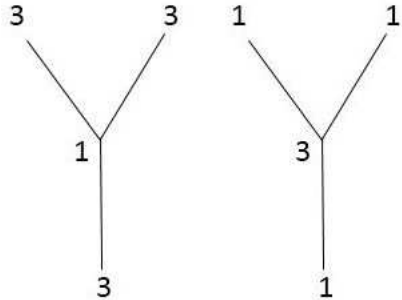}
    \end{center}
    \caption{Configurations for $A_6$} \label{case6}
\end{figure}

Let us consider the quotient group $G_2/ G_2^{(2)}=\{H_0, H_1\},$ where
$$
H_0=G_2^{(2)}=\{x\in G_2 : |x|\ \ \mbox {is even}\},\ \ H_1=\{x\in
G_2 : |x|\ \ \mbox {is odd}\}.
$$

Using configuration given by Figure \ref{case6}, one can construct
configuration define by:

        \begin{equation}\label{5}
        \v_{i,j}^{(6)}(x)=\left\{
        \begin{array}{ll}
        i &,x \in H_0,\\
        j &,x \in H_1.
        \end{array}\right.
        \end{equation}
        where $|i-j|=2$ and $ i,j\in \Phi$.

Configurations $\v_{i,j}^{(6)}(x)$ are ground states. Indeed, let
$b$ be any initial ball from $M$ and $C_b\in H_0$ then
$C_{{b}{\downarrow}}, x_b, y_b$ belongs to class $H_1$, i.e.
$(\v^{(5)}_{i,j})_b\in C_5.$ If $C_b\in H_1$ then
$C_{{b}{\downarrow}}, x_b, y_b$ belongs to class $H_0$, i.e.
$(\v^{(5)}_{i,j})_b\in C_5.$

    Note that  a number of the configurations $\v^{(5)}_{i,j}, |i-j|=2, i,j\in \Phi$ is two.
    Theorem is proved.
\end{proof}

\begin{thm}\label{th33}
    Let $(\ab,\bb,\cb, J) \in A_7$, then there is not any ground states.
\end{thm}

\begin{proof}
Let $(\ab,\bb,\cb, J) \in A_7$, then one can see that for this triple, the minimal value is $3\ab/2+J$.
Let $\v$ be a ground state, i.e. for any $b\in
M$, one has $\v_b(x)\in C_7$. Then
$\v(C_b)=2$, since $U(\v_b)=3\ab/2+J$. From $\v_b(x)\in C_7$ we conclude that all of the following variables
$$|\v(C_b)-\v(C_{{b}{\downarrow}})|, \  |\v(C_b)-\v(x_b)|, \ 
|\v(C_b)-\v(y_b)|$$ must be equal to 1, i.e. for example $\v(x_b)$ is not
equal to two.  If we consider of unit ball $b'$ with center $x_b$,
then $U(\v_b')\neq (3\ab)/2+J$. Consequently there is not any ground state. \end{proof}

\begin{thm}\label{th4}
    Let $(\ab,\bb,\cb, J) \in A_l, l=8, 9, 10$, then there is not ground states.
\end{thm}
\begin{proof}

The proof of this theorem is similar to proof of the Theorem
    \ref{th33}.
\end{proof}

\begin{thm}\label{th4}
    Let $(\ab,\bb,\cb, J) \in A_{11}$, then there are four $G_2^{(2)}-$periodic ground states.
\end{thm}

\begin{proof}

    Let $(\ab,\bb,\cb,J) \in A_6$, then one can see that for this triple, the minimal value is $(3\bb)/2+3J$.

Consider the quotient group $G_2/ G_2^{(2)}=\{H_0, H_1\},$ where
$$
H_0=G_2^{(2)}=\{x\in G_2 : |x|\ \ \mbox {is even}\},\ \ H_1=\{x\in
G_2 : |x|\ \ \mbox {is odd}\}.
$$

Let
        \begin{equation}\label{5}
        \v_{i,j}^{(11)}(x)=\left\{
        \begin{array}{ll}
        i &,x \in H_0,\\
        j &,x \in H_1.
        \end{array}\right.
        \end{equation}
        where $|i-j|=1$ and $ i,j\in \Phi$.

Configurations $\v_{i,j}^{(11)}(x)$ are ground states. Really, let
$b$ is any initial ball from $M$ and $C_b\in H_0$ then
$C_{{b}{\downarrow}}, x_b, y_b$ are belong to class $H_1$, i.e.
$(\v^{(11)}_{i,j})_b\in C_{11}.$ If $C_b\in H_1$ then
$C_{{b}{\downarrow}}, x_b, y_b$ belong to class $H_0$, i.e.
$(\v^{(11)}_{i,j})_b\in C_{11}.$

    Note that number of the configurations $\v^{(11)}_{i,j}, |i-j|=1, i,j\in \Phi$ is four.
    Theorem is proved.

\end{proof}

Let  $G_2^{(4)}=\{x \in G_2: |x|-even,\ w_x(a_1)-even\}$.
Note that $G_2^{(4)}$ is a normal subgroup of index four.
We consider the following quotient group $G_2/G_2^{(4)}=\{H_0, H_1,
H_2, H_4\},$ where
$$
\begin{array}{llll}
           H_0=G_2^{(4)}\\
            H_1=\{x \in G_2:|x|-even,\ w_x(a_1)-odd\},\\
            H_2=\{x \in G_2:|x|-odd,\ w_x(a_1)-even\},\\
            H_3=\{x \in G_2:|x|-odd,\ w_x(a_1)-odd\}.
        \end{array}$$

\begin{thm}\label{th}
    Let $(\ab,\bb,\cb) \in A_{12}$, then the following statements hold.
    \begin{itemize}
        \item[(i)] there is uncountable number of ground states;
        \item[(ii)] there exist four $G_2^{(4)}-$periodic ground states.
    \end{itemize}
\end{thm}

\begin{proof}
$(i)$. Proof of this statement is similar to proof of statement $(i)$ of the Theorem \ref{th2}.

$(ii)$. Let

        \begin{equation}\label{12}
        \v_{i,j}^{(12)}(x)=\left\{
        \begin{array}{ll}
        i, \ \mbox{if} \ x \in H_0\cup H_3,\\
        j, \ \mbox{if} \ x \in H_1\cup H_2,
        \end{array}\right.
        \end{equation}
        be the $G_2^{(4)}-$periodic configuration, where $|i-j|=1$.
        We shall prove that the periodic configuration $\v_{i,j}^{(12)}(x)$ is a periodic ground states. Let $b\in M$ is arbitrary unit ball and $C_b\in H_0$, then it is easy to see that $|\{C_{{b}{\downarrow}}, x_b,y_b\}\cap H_0|=0$, $|\{C_{{b}{\downarrow}}, x_b,y_b\}\cap H_1|=0$, $|\{C_{{b}{\downarrow}}, x_b,y_b\}\cap H_2|=2$ and $|\{C_{{b}{\downarrow}}, x_b,y_b\}\cap H_3|=1$. In this case by \ref{12} may be the following:\\
        1) $\v_{i,j}^{(12)}(C_b)=i, \v_{i,j}^{(12)}(C_{{b}{\downarrow}})=i, \v_{i,j}^{(12)}(x_b)=j, \v_{i,j}^{(12)}(y_b)=j$;\\
        2) $\v_{i,j}^{(12)}(C_b)=i, \v_{i,j}^{(12)}(C_{{b}{\downarrow}})=j, \v_{i,j}^{(12)}(x_b)=i, \v_{i,j}^{(12)}(y_b)=j$;\\
        3) $\v_{i,j}^{(12)}(C_b)=i, \v_{i,j}^{(12)}(C_{{b}{\downarrow}})=j, \v_{i,j}^{(12)}(x_b)=j, \v_{i,j}^{(12)}(y_b)=i$;\\
        In all cases $U((\v_{i,j}^{(12)}(x))_b)=(2\bb+\cb)/2+J$.

      If $C_b\in H_1$, then it is easy to see that $|\{C_{{b}{\downarrow}}, x_b,y_b\}\cap H_0|=0$, $|\{C_{{b}{\downarrow}}, x_b,y_b\}\cap H_1|=0$, $|\{C_{{b}{\downarrow}}, x_b,y_b\}\cap H_2|=1$ and $|\{C_{{b}{\downarrow}}, x_b,y_b\}\cap H_3|=2$. In this case by \ref{12} may be the following:\\
        1) $\v_{i,j}^{(12)}(C_b)=j, \v_{i,j}^{(12)}(C_{{b}{\downarrow}})=j, \v_{i,j}^{(12)}(x_b)=i, \v_{i,j}^{(12)}(y_b)=i$;\\
        2) $\v_{i,j}^{(12)}(C_b)=j, \v_{i,j}^{(12)}(C_{{b}{\downarrow}})=i, \v_{i,j}^{(12)}(x_b)=i, \v_{i,j}^{(12)}(y_b)=j$;\\
        3) $\v_{i,j}^{(12)}(C_b)=j, \v_{i,j}^{(12)}(C_{{b}{\downarrow}})=i, \v_{i,j}^{(12)}(x_b)=j, \v_{i,j}^{(12)}(y_b)=i$;\\
        In all cases $U((\v_{i,j}^{(12)}(x))_b)=(2\bb+\cb)/2+J$.

        If $C_b\in H_2$, then it is easy to see that $|\{C_{{b}{\downarrow}}, x_b,y_b\}\cap H_0|=2$, $|\{C_{{b}{\downarrow}}, x_b,y_b\}\cap H_1|=1$, $|\{C_{{b}{\downarrow}}, x_b,y_b\}\cap H_2|=0$ and $|\{C_{{b}{\downarrow}}, x_b,y_b\}\cap H_3|=0$. In this case by \ref{12} may be the following:\\
        1) $\v_{i,j}^{(12)}(C_b)=j, \v_{i,j}^{(12)}(C_{{b}{\downarrow}})=j, \v_{i,j}^{(12)}(x_b)=i, \v_{i,j}^{(12)}(y_b)=i$;\\
        2) $\v_{i,j}^{(12)}(C_b)=j, \v_{i,j}^{(12)}(C_{{b}{\downarrow}})=i, \v_{i,j}^{(12)}(x_b)=i, \v_{i,j}^{(12)}(y_b)=j$;\\
        3) $\v_{i,j}^{(12)}(C_b)=j, \v_{i,j}^{(12)}(C_{{b}{\downarrow}})=i, \v_{i,j}^{(12)}(x_b)=j, \v_{i,j}^{(12)}(y_b)=i$;\\
        In all cases $U((\v_{i,j}^{(12)}(x))_b)=(2\bb+\cb)/2+J$.

        If $C_b\in H_3$, then it is easy to see that $|\{C_{{b}{\downarrow}}, x_b,y_b\}\cap H_0|=1$, $|\{C_{{b}{\downarrow}}, x_b,y_b\}\cap H_1|=2$, $|\{C_{{b}{\downarrow}}, x_b,y_b\}\cap H_2|=0$ and $|\{C_{{b}{\downarrow}}, x_b,y_b\}\cap H_3|=0$. In this case by \ref{12} may be the following:\\
        1) $\v_{i,j}^{(12)}(C_b)=i, \v_{i,j}^{(12)}(C_{{b}{\downarrow}})=i, \v_{i,j}^{(12)}(x_b)=j, \v_{i,j}^{(12)}(y_b)=j$;\\
        2) $\v_{i,j}^{(12)}(C_b)=i, \v_{i,j}^{(12)}(C_{{b}{\downarrow}})=j, \v_{i,j}^{(12)}(x_b)=i, \v_{i,j}^{(12)}(y_b)=j$;\\
        3) $\v_{i,j}^{(12)}(C_b)=i, \v_{i,j}^{(12)}(C_{{b}{\downarrow}})=j, \v_{i,j}^{(12)}(x_b)=j, \v_{i,j}^{(12)}(y_b)=i$;\\

        In all cases $U((\v_{i,j}^{(12)}(x))_b)=(2\bb+\cb)/2+J$, i.e. $(\v_{i,j}^{(12)})_b \in C_{12}, \forall b\in M$  consequently periodic configuration $\v_{i,j}^{(12)}(x)$ is ground states on the set $A_{12}.$

\end{proof}

\section{Gibbs measures of the $\lambda$-model with competing Potts interactions}\label{Gibbs measures}

In this section, we define a notion of Gibbs measure corresponding to
the  $\lambda$ model with competing Potts interactions on an arbitrary order Cayley tree. We
propose a new kind of construction of Gibbs measures corresponding
to the model.

Below, for the sake of simplicity, we will consider a semi-infinite
Cayley tree $\tau_+^k$ of order $k$, i.e. an infinite graph
without cycles with $k+1$ edges issuing from each vertex except for
$x^0$ which has only $k$ edges.

In what follows, for the sake of simplicity of calculations, we consider the model where the spin takes values in the set
$\Phi=\{\eta_1,\eta_2,\eta_3\}$. 
Here $\eta_1,\eta_2,\eta_3$ are vectors in $\br^2$ such that 
$$
\eta_i\cdot\eta_j=
\left\{
\begin{array}{ll}
\ \ 1, \ \ \ i=j\\[2mm]
-\frac{1}{2}, \ \ i\neq j\\
\end{array}
\right.
$$
We racall that the set of configurations on $V$(resp.  $V_n$ and $W_n$) coincides with
$\Omega=\Phi^{V}$ (resp. $\Omega_{V_n}=\Phi^{V_n},\ \
\Omega_{W_n}=\Phi^{W_n}$). One can see that
$\Om_{V_n}=\Om_{V_{n-1}}\times\Om_{W_n}$. Using this, for given
configurations $\s_{n-1}\in\Om_{V_{n-1}}$ and $\w\in\Om_{W_{n}}$ we
define their concatenations  by
$$
(\s_{n-1}\vee\w)(x)= \left\{
\begin{array}{ll}
\s_{n-1}(x), \ \ \textrm{if} \ \  x\in V_{n-1},\\
\w(x), \ \ \ \ \ \ \textrm{if} \ \ x\in W_n.\\
\end{array}
\right.
$$
It is clear that $\s_{n-1}\vee\w\in \Om_{V_n}$.

In this section, for the sake of simplicity,  the $\lambda$ model with competing Potts interactions is given by the following Hamiltonian 
\begin{equation}\label{ham}
H(\sigma)=-J\sum_{\widetilde{>x,y<}}\delta_{\sigma(x)\sigma(y)} -
\sum_{<x,y>}\lambda(\sigma(x)\sigma(y)),
\end{equation}

Assume that  $\h: (V\setminus\{x^0\})\times
(V\setminus\{x^0\})\times \Phi\times\Phi\to \mathbb{R}^9$ is a
mapping, i.e.
$$\h_{xy,uv}=\bigg(h_{xy,\eta_1\eta_1},h_{xy,\eta_1\eta_2},h_{xy,\eta_1\eta_3},h_{xy,\eta_2\eta_1},h_{xy,\eta_2\eta_2},h_{xy,\eta_2\eta_3},h_{xy,\eta_3\eta_1},h_{xy,\eta_3\eta_2},h_{xy,\eta_3\eta_3}\bigg),$$
where $h_{xy,uv}\in \mathbb{R}$, $u,v\in\Phi$, and $x,y\in
V\setminus\{x^{(0)}\}$.

Now, we define the Gibbs measure with memory of length 2 on the
Cayley tree as follows:
\begin{equation}\label{mu}
\m^{(n)}_{\h}(\s)=\frac{1}{Z_{n}}\exp[-\beta H_n(\s)+\sum_{x\in
    W_{n-1}}\sum_{y\in
    S(x)}\sigma(x)\sigma(y)\h_{xy,\sigma(x)\sigma(y)}].
\end{equation}

Here, $\beta=\frac{1}{kT}$, $\sigma\in\Omega_{V_n}$ and $Z_n$ is the
corresponding to partition function
\begin{equation}\label{Zn}
Z_{n}=\sum\limits_{\sigma_n\in \Omega_{V_n}}\exp[-\beta
H(\s_n)+\sum_{x\in W_{n-1}}\sum_{y\in
    S(x)}\sigma(x)\sigma(y)\h_{xy,\sigma(x)\sigma(y)}].
\end{equation}

In order to construct an infinite volume distribution with given finite-dimensional distributions, we would like to find a probability measure $\m$ on
$\Om$ with given conditional probabilities $\m_{\h}^{(n)}$, i.e.
\begin{equation}\label{CM}
\m(\s\in\Om: \s|_{V_n}=\s_n)=\m^{(n)}_{\h}(\s_n), \ \ \
\textrm{for all} \ \ \s_n\in\Om_{V_n}, \ n\in\bn.
\end{equation}
If the measures $\{\m^{(n)}_{\h}\}$ are \textit{compatible}, i.e.
\begin{equation}\label{comp}
\sum_{\w\in\Om_{W_n}}\m^{(n)}_{\h}(\s\vee\w)=\m^{(n-1)}_{\h}(\s), \
\ \ \textrm{for any} \ \ \s\in\Om_{V_{n-1}},
\end{equation}
then according to the Kolmogorov's theorem there exists a unique
measure $\m_{\h}$ defined on $\Om$ with a required condition
\eqref{CM}. Such a measure $\m_{\h}$ is said to be {Gibbs measure}
corresponding to the model. Note that a general theory of Gibbs
measures has been developed in \cite{1,8}.

The next statement describes the conditions on the boundary fields
$\h$ guaranteeing the compatibility of the distributions
$\{\m^{(n)}_\h\}$ .

\begin{thm}\label{theorem1}
    The measures  $\m^{(n)}_\h$, $n=1,2,...,$ in \eqref{mu} are
    compatible iff for any $x,y\in V$ the following equations hold:
    \begin{equation}\label{necessary}
    \left\{
    \begin{array}{ll}
    e^{-\frac{1}{2} h_{xy,\e_1\e_2}-h_{xy,\e_1\e_1}}=\prod\limits_{z\in
        S(y)}\frac{\exp[-\frac{1}{2}\h_{yz,\e_2\e_1}]bd+\exp[\h_{yz,\e_2\e_2}]c+\exp[-\frac{1}{2}\h_{yz,\e_2\e_3}]b}
    {\exp[\h_{yz,\e_1\e_1}]cd+\exp[-\frac{1}{2}\h_{yz,\e_1\e_2}]b+\exp[-\frac{1}{2}\h_{yz,\e_1\e_3}]a}\\[2mm]
    e^{-\frac{1}{2} h_{xy,\e_1\e_3}-h_{xy,\e_1\e_1}}=\prod\limits_{z\in
        S(y)}\frac{\exp[-\frac{1}{2}\h_{yz,\e_3\e_1}]ad+\exp[-\frac{1}{2}\h_{yz,\e_3\e_2}]b+\exp[\h_{yz,\e_3\e_3}]c}
    {\exp[\h_{yz,\e_1\e_1}]cd+\exp[-\frac{1}{2}\h_{yz,\e_1\e_2}]b+\exp[-\frac{1}{2}\h_{yz,\e_1\e_3}]a}\\[2mm]
    e^{-\frac{1}{2} h_{xy,\e_2\e_1}-h_{xy,\e_1\e_1}}=\prod\limits_{z\in
        S(y)}\frac{\exp[\h_{yz,\e_1\e_1}]c+\exp[-\frac{1}{2}\h_{yz,\e_1\e_2}]bd+\exp[-\frac{1}{2}\h_{yz,\e_1\e_3}]a}
    {\exp[\h_{yz,\e_1\e_1}]cd+\exp[-\frac{1}{2}\h_{yz,\e_1\e_2}]b+\exp[-\frac{1}{2}\h_{yz,\e_1\e_3}]a}\\[2mm]
    e^{h_{xy,\e_2\e_2}-h_{xy,\e_1\e_1}}=\prod\limits_{z\in
        S(y)}\frac{\exp[-\frac{1}{2}\h_{yz,\e_2\e_1}]b+\exp[\h_{yz,\e_2\e_2}]cd+\exp[-\frac{1}{2}\h_{yz,\e_2\e_3}]b}
    {\exp[\h_{yz,\e_1\e_1}]cd+\exp[-\frac{1}{2}\h_{yz,\e_1\e_2}]b+\exp[-\frac{1}{2}\h_{yz,\e_1\e_3}]a}\\[2mm]
    e^{-\frac{1}{2}h_{xy,\e_2\e_3}-h_{xy,\e_1\e_1}}=\prod\limits_{z\in  S(y)}\frac{\exp[-\frac{1}{2}\h_{yz,\e_3\e_1}]a+\exp[-\frac{1}{2}\h_{yz,\e_3\e_2}]bd+\exp[\h_{yz,\e_3\e_3}]c}
    {\exp[\h_{yz,\e_1\e_1}]cd+\exp[-\frac{1}{2}\h_{yz,\e_1\e_2}]b+\exp[-\frac{1}{2}\h_{yz,\e_1\e_3}]a}\\[2mm]
    e^{-\frac{1}{2}h_{xy,\e_3\e_1}-h_{xy,\e_1\e_1}}=\prod\limits_{z\in  S(y)}\frac{\exp[\h_{yz,\e_1\e_1}]c+\exp[-\frac{1}{2}\h_{yz,\e_1\e_2}]b+\exp[-\frac{1}{2}\h_{yz,\e_1\e_3}]ad}
    {\exp[\h_{yz,\e_1\e_1}]cd+\exp[-\frac{1}{2}\h_{yz,\e_1\e_2}]b+\exp[-\frac{1}{2}\h_{yz,\e_1\e_3}]a}\\[2mm]
    e^{-\frac{1}{2}h_{xy,\e_3\e_2}-h_{xy,\e_1\e_1}}=\prod\limits_{z\in  S(y)}\frac{\exp[-\frac{1}{2}\h_{yz,\e_2\e_1}]b+\exp[\h_{yz,\e_2\e_2}]c+\exp[-\frac{1}{2}\h_{yz,\e_2\e_3}]bd}
    {\exp[\h_{yz,\e_1\e_1}]cd+\exp[-\frac{1}{2}\h_{yz,\e_1\e_2}]b+\exp[-\frac{1}{2}\h_{yz,\e_1\e_3}]a}\\[2mm]
    e^{h_{xy,\e_3\e_3}-h_{xy,\e_1\e_1}}=\prod\limits_{z\in  S(y)}\frac{\exp[-\frac{1}{2}\h_{yz,\e_3\e_1}]a+\exp[\h_{yz,\e_3\e_2}]b+\exp[\h_{yz,\e_3\e_3}]cd}
    {\exp[\h_{yz,\e_1\e_1}]cd+\exp[-\frac{1}{2}\h_{yz,\e_1\e_2}]b+\exp[-\frac{1}{2}\h_{yz,\e_1\e_3}]a}\\

    \end{array}
    \right.
    \end{equation}
    where $a=\exp(\beta \ab)$, $b=\exp(\beta \bb)$, $c=\exp(\beta \cb)$ and $d=\exp(\beta J)$.
\end{thm}
\begin{proof}
    {\sc Necessity}. From \eqref{comp}, we have
    \begin{eqnarray}\label{comp111}
    &&L_n\sum\limits_{\eta\in \Omega_{W_{n}}}\exp[-\beta H_n(\s\vee
    \eta)+
    \sum\limits_{x\in W_{n-1}}\sum\limits_{y\in
        S(x)}\sigma(x)\sigma(y)\h_{xy,\sigma(x)\sigma(y)}]\nonumber\\[2mm]
    &=&\exp [-\beta H_n(\s)+\sum\limits_{x\in W_{n-2}}\sum\limits_{y\in
        S(x)}\sigma(x)\sigma(y)\h_{xy,\sigma(x)\sigma(y)}],
    \end{eqnarray}
    where $L_n=\frac{Z_{n-1}}{Z_n}$.

    For $\s\in V_{n-1}$ and $\eta \in W_{n}$, we rewrite the Hamiltonian
    as follows:

    \begin{eqnarray}\label{ham1}
    H_n(\s\vee\eta)
    &=&\sum\limits_{<x,y>\in V_{n-1}}\lambda(\sigma(x)\sigma(y)) +
    \sum\limits_{x\in W_{n-1}}\sum\limits_{y\in
        S(x)}\lambda(\sigma(x)\eta(y))\nonumber \\
    \nonumber
    &&-J\sum\limits_{>x,y<\in V_{n-1}}\delta_{\sigma(x)\sigma(y)} -J
    \sum\limits_{x\in
        W_{n-2}}\sum\limits_{z\in S^2(x)}\delta_{\sigma(x)\eta(z)}\\
    &=&H_{n-1}(\s)+\sum\limits_{x\in W_{n-1}}\sum\limits_{y\in
        S(x)}\lambda(\sigma(x)\eta(y))-J \sum\limits_{x\in
        W_{n-2}}\sum\limits_{z\in S^2(x)}\delta_{\sigma(x)\eta(z)}.
    \end{eqnarray}

    Therefore, the last equality with \eqref{comp111} implies
    \begin{eqnarray}\label{Kolmogorov1}
    &&L_n\sum\limits_{\eta\in \Omega_{W_{n}}}\exp[-\beta
    H_{n-1}(\s)-\beta \sum\limits_{x\in W_{n-1}}\sum\limits_{y\in
        S(x)}\lambda(\sigma(x)\eta(y))\nonumber\\
    \nonumber &+&
    J\sum\limits_{x\in W_{n-2}}\sum\limits_{z\in
        S^2(x)}\delta_{\sigma(x)\eta(z)}+ \sum\limits_{x\in W_{n-1}}\sum\limits_{y\in
        S(x)}\sigma(x)\sigma(y)\h_{xy,\sigma(x)\eta(y)}]\\
    &=&\exp [-\beta H_{n-1}(\s)+\sum\limits_{x\in
        W_{n-2}}\sum\limits_{y\in
        S(x)}\sigma(x)\sigma(y)\h_{xy,\sigma(x)\sigma(y)}],
    \end{eqnarray}
    %
    Hence, one gets
    \begin{eqnarray*}\label{Kolmogorov2}
        &&L_n\prod\limits_{x\in W_{n-2}}\prod\limits_{y\in S(x)}\prod\limits_{z\in S(y)}\sum\limits_{\eta(z)\in \{\eta_{1},\eta_{2},\eta_{3}\}}\exp[-\beta \lambda(\sigma(x),\eta(z))+\beta J\delta_{\sigma(x)\eta(z)}+\eta(z)\sigma(y)\eta(z)\h_{yz,\sigma(y)\eta(z)}]\\\nonumber
        &=&\prod\limits_{x\in W_{n-2}}\prod\limits_{y\in S(x)}\exp [\sigma(x)\sigma(y)\h_{xy,\sigma(x)\sigma(y)}].
    \end{eqnarray*}

    Let us fix $<x,y>$. Then considering all values of $\sigma(x),
    \sigma(y)\in \{\eta_1,\eta_2,\eta_3\}$, from \eqref{Kolmogorov1}, we obtain

    \begin{align}\label{Kolmogorov12}
        e^{-\frac{1}{2}h_{xy,\eta_{1}\eta_{2}}-h_{xy,\eta_{1}\eta_{1}}}&=\prod\limits_{z\in S(y)}\frac{\sum\limits_{\eta(z)\in \{\eta_{1},\eta_{2},\eta_{3}\}}\exp[-\beta \lambda(\eta_{1},\eta(z))+\beta J\delta_{\eta_{1}\eta(z)}+\eta_{2}\eta(z)\h_{yz,\eta_{2}\eta(z)}]}{\sum\limits_{\eta(z)\in \{\eta_{1},\eta_{2},\eta_{3}\}}\exp[-\beta \lambda(\eta_{1},\eta(z))+\beta J\delta_{\eta_{1}\eta(z)}+\eta_{1}\eta(z)\h_{yz,\eta_{1}\eta(z)}]}\\\nonumber
        &=\frac{\Lambda(\eta_{1},\eta_{2})}{\Lambda(\eta_{1},\eta_{1})}
    \end{align}

    \begin{align}\label{Kolmogorov13}
        e^{-\frac{1}{2}h_{xy,\eta_{1}\eta_{3}}-h_{xy,\eta_{1}\eta_{1}}}&=\prod\limits_{z\in S(y)}\frac{\sum\limits_{\eta(z)\in \{\eta_{1},\eta_{2},\eta_{3}\}}\exp[-\beta \lambda(\eta_{1},\eta(z))+\beta J\delta_{\eta_{1}\eta(z)}+\eta_{3}\eta(z)\h_{yz,\eta_{3}\eta(z)}]}{\sum\limits_{\eta(z)\in \{\eta_{1},\eta_{2},\eta_{3}\}}\exp[-\beta \lambda(\eta_{1},\eta(z))+\beta J\delta_{\eta_{1}\eta(z)}+\eta_{1}\eta(z)\h_{yz,\eta_{1}\eta(z)}]}\\\nonumber
        &=\frac{\Lambda(\eta_{1},\eta_{3})}{\Lambda(\eta_{1},\eta_{1})}
    \end{align}

    \begin{align}\label{Kolmogorov21}
        e^{-\frac{1}{2}h_{xy,\eta_{2}\eta_{1}}-h_{xy,\eta_{1}\eta_{1}}}&=\prod\limits_{z\in S(y)}\frac{\sum\limits_{\eta(z)\in \{\eta_{1},\eta_{2},\eta_{3}\}}\exp[-\beta \lambda(\eta_{2},\eta(z))+\beta J\delta_{\eta_{2}\eta(z)}+\eta_{1}\eta(z)\h_{yz,\eta_{1}\eta(z)}]}{\sum\limits_{\eta(z)\in \{\eta_{1},\eta_{2},\eta_{3}\}}\exp[-\beta \lambda(\eta_{1},\eta(z))+\beta J\delta_{\eta_{1}\eta(z)}+\eta_{1}\eta(z)\h_{yz,\eta_{1}\eta(z)}]}\\\nonumber
        &=\frac{\Lambda(\eta_{2},\eta_{1})}{\Lambda(\eta_{1},\eta_{1})}
    \end{align}

    \begin{align}\label{Kolmogorov22}
        e^{h_{xy,\eta_{2}\eta_{2}}-h_{xy,\eta_{1}\eta_{1}}}&=\prod\limits_{z\in S(y)}\frac{\sum\limits_{\eta(z)\in \{\eta_{1},\eta_{2},\eta_{3}\}}\exp[-\beta \lambda(\eta_{2},\eta(z))+\beta J\delta_{\eta_{2}\eta(z)}+\eta_{2}\eta(z)\h_{yz,\eta_{2}\eta(z)}]}{\sum\limits_{\eta(z)\in \{\eta_{1},\eta_{2},\eta_{3}\}}\exp[-\beta \lambda(\eta_{1},\eta(z))+\beta J\delta_{\eta_{1}\eta(z)}+\eta_{1}\eta(z)\h_{yz,\eta_{1}\eta(z)}]}\\\nonumber
        &=\frac{\Lambda(\eta_{2},\eta_{2})}{\Lambda(\eta_{1},\eta_{1})}
    \end{align}

    \begin{align}\label{Kolmogorov23}
        e^{-\frac{1}{2}h_{xy,\eta_{2}\eta_{3}}-h_{xy,\eta_{1}\eta_{1}}}=&\prod\limits_{z\in S(y)}\frac{\sum\limits_{\eta(z)\in \{\eta_{1},\eta_{2},\eta_{3}\}}\exp[-\beta \lambda(\eta_{2},\eta(z))+\beta J\delta_{\eta_{2}\eta(z)}+\eta_{3}\eta(z)\h_{yz,\eta_{3}\eta(z)}]}{\sum\limits_{\eta(z)\in \{\eta_{1},\eta_{2},\eta_{3}\}}\exp[-\beta \lambda(\eta_{1},\eta(z))+\beta J\delta_{\eta_{1}\eta(z)}+\eta_{1}\eta(z)\h_{yz,\eta_{1}\eta(z)}]}\\\nonumber
        &=\frac{\Lambda(\eta_{2},\eta_{3})}{\Lambda(\eta_{1},\eta_{1})}
    \end{align}

    \begin{align}\label{Kolmogorov31}
        e^{-\frac{1}{2}h_{xy,\eta_{3}\eta_{1}}-h_{xy,\eta_{1}\eta_{1}}}&=\prod\limits_{z\in S(y)}\frac{\sum\limits_{\eta(z)\in \{\eta_{1},\eta_{2},\eta_{3}\}}\exp[-\beta \lambda(\eta_{3},\eta(z))+\beta J\delta_{\eta_{3}\eta(z)}+\eta_{1}\eta(z)\h_{yz,\eta_{1}\eta(z)}]}{\sum\limits_{\eta(z)\in \{\eta_{1},\eta_{2},\eta_{3}\}}\exp[-\beta \lambda(\eta_{1},\eta(z))+\beta J\delta_{\eta_{1}\eta(z)}+\eta_{1}\eta(z)\h_{yz,\eta_{1}\eta(z)}]}\\\nonumber
        &=\frac{\Lambda(\eta_{3},\eta_{1})}{\Lambda(\eta_{1},\eta_{1})}
    \end{align}

    \begin{align}\label{Kolmogorov32}
        e^{-\frac{1}{2}h_{xy,\eta_{3}\eta_{2}}-h_{xy,\eta_{1}\eta_{1}}}&=\prod\limits_{z\in S(y)}\frac{\sum\limits_{\eta(z)\in \{\eta_{1},\eta_{2},\eta_{3}\}}\exp[-\beta \lambda(\eta_{3},\eta(z))+\beta J\delta_{\eta_{3}\eta(z)}+\eta_{2}\eta(z)\h_{yz,\eta_{2}\eta(z)}]}{\sum\limits_{\eta(z)\in \{\eta_{1},\eta_{2},\eta_{3}\}}\exp[-\beta \lambda(\eta_{1},\eta(z))+\beta J\delta_{\eta_{1}\eta(z)}+\eta_{1}\eta(z)\h_{yz,\eta_{1}\eta(z)}]}\\\nonumber
        &=\frac{\Lambda(\eta_{3},\eta_{2})}{\Lambda(\eta_{1},\eta_{1})}
    \end{align}

    \begin{align}\label{Kolmogorov33}
        e^{h_{xy,\eta_{3}\eta_{3}}-h_{xy,\eta_{1}\eta_{1}}}&=\prod\limits_{z\in S(y)}\frac{\sum\limits_{\eta(z)\in \{\eta_{1},\eta_{2},\eta_{3}\}}\exp[-\beta \lambda(\eta_{3},\eta(z))+\beta J\delta_{\eta_{3}\eta(z)}+\eta_{3}\eta(z)\h_{yz,\eta_{3}\eta(z)}]}{\sum\limits_{\eta(z)\in \{\eta_{1},\eta_{2},\eta_{3}\}}\exp[-\beta \lambda(\eta_{1},\eta(z))+\beta J\delta_{\eta_{1}\eta(z)}+\eta_{1}\eta(z)\h_{yz,\eta_{1}\eta(z)}]}\\\nonumber
        &=\frac{\Lambda(\eta_{3},\eta_{3})}{\Lambda(\eta_{1},\eta_{1})}
    \end{align}
    These equations imply the desired ones.\\

    {\sc Sufficiency}. Now we assume that the system of equations
    \eqref{necessary} is valid, then one finds

    $$e^{\sigma(x)\sigma(y)h_{xy,\sigma(x)\sigma(y)}}D(x,y)=\prod\limits_{z\in S(y)}\sum\limits_{\eta(z)\in \{\eta_1,\eta_2,\eta_3\}}
    \exp[\sigma(y)\eta(z)\h_{yz,\sigma(y)\eta(z)}+\beta
    \eta(z)(\sigma(y)+J \sigma(x))],$$ for some constant $D(x,y)$
    depending on $x$ and $y$.

    From the last equality, we obtain
    \begin{eqnarray}\label{Kolmogorov5}
    &&\prod\limits_{x\in W_{n-2}}\prod\limits_{y\in S(x)}D(x,y)e^{\sigma(x)\sigma(y)h_{xy,\sigma(x)\sigma(y)}}\\\nonumber
    &=&\prod\limits_{x\in W_{n-2}}\prod\limits_{y\in S(x)}\prod\limits_{z\in S(y)}\sum\limits_{\eta(z)\in \{\eta_1,\eta_2,\eta_3\}}e^{[\sigma(y)\eta(z)\h_{yz,\sigma(y)\eta(z)}+\beta \eta(z)(\sigma(y)+ J \sigma(x))]}.
    \end{eqnarray}
    Multiply both sides of the equation \eqref{Kolmogorov5} by $e^{-\beta H_{n-1}(\sigma)}$ and denoting

    $$U_{n-1}=\prod\limits_{x\in W_{n-2}}\prod\limits_{y\in S(x)}D(x,y),$$
    from \eqref{Kolmogorov5}, one has

    \begin{eqnarray*}
        &&U_{n-1}e^{-\beta H_{n-1}(\sigma)+\sum\limits_{x\in W_{n-2}}
            \sum\limits_{y\in S(x)}\sigma(x)\sigma(y)h_{xy,\sigma(x)\sigma(y)}}\\
        &=&\prod\limits_{x\in W_{n-2}}\prod\limits_{y\in S(x)}
        \prod\limits_{z\in S(y)}e^{-\beta H_{n-1}(\sigma)}
        \sum\limits_{\eta(z)\in \{\eta_1,\eta_2,\eta_3\}}e^{[\sigma(y)\eta(z)\h_{yz,\sigma(y)\eta(z)}+\beta \eta(z)(\sigma(y)+ J \sigma(x))]}.
    \end{eqnarray*}
    which yields
    $$
    U_{n-1}Z_{n-1}\m^{(n-1)}_\h(\sigma)=\sum\limits_{\eta}e^{-\beta
        H_{n}(\sigma\vee\eta) +\sum\limits_{x\in W_{n-2}}\sum\limits_{y\in
            S(x)}\sigma(x)\sigma(y)h_{xy,\sigma(x)\sigma(y)}}.
    $$
    This means
    \begin{eqnarray}\label{eq4}
    U_{n-1}Z_{n-1}\m^{(n-1)}_\h(\sigma)=Z_{n}\sum\limits_{\eta}\m^{(n)}_\h(\sigma\vee\eta).
    \end{eqnarray}
    As $\m^{(n)}_\h$ ($n\geq 1$) is a probability  measure, i.e.
    $$
    \sum\limits_{\sigma\in \{\eta_1,\eta_2,\eta_3\}^{V_{n-1}}}\m^{(n-1)}_\h(\sigma)
    =\sum\limits_{\sigma\in \{\eta_1,\eta_2,\eta_3\}^{V_{n-1}}}\sum\limits_{\eta \in
        \{\eta_1,\eta_2,\eta_3\}^{W_{n}}}\m^{(n)}_\h(\sigma\vee\eta)=1.
    $$
    From these equalities and \eqref{eq4} we have
    $Z_{n}=U_{n-1}Z_{n-1}$. This with \eqref{eq4} implies that
    \eqref{comp} holds. The proof is complete.
\end{proof}
According to Theorem \ref{theorem1} the problem of describing the
Gibbs measures is reduced to the descriptions of the solutions of
the functional equations \eqref{necessary}.

\begin{cor}\label{compatibility}
    The measures $\m^{(n)}_\h$, $ n=1,2,\dots$ satisfy the compatibility
    condition \eqref{comp} if and only if for any $n\in \bn$ the
    following equation holds:
    \begin{equation}\label{canonic_u}
    \left\{\begin{array}{ll}
    u_{xy,1}=\prod\limits_{z\in
        S(y)}\frac{u_{yz,3}bd+u_{yz,4}c+u_{yz,5}b}
    {u_{yz,1}b+u_{yz,2}a+cd},\ \ \    u_{xy,2}=\prod\limits_{z\in
        S(y)}\frac{u_{yz,6}ad+u_{yz,7}b+u_{yz,8}c}
    {u_{yz,1}b+u_{yz,2}a+cd},\\[4mm]
    u_{xy,3}=\prod\limits_{z\in
        S(y)}\frac{c+u_{yz,1}bd+u_{yz,2}a}
    {u_{yz,1}b+u_{yz,2}a+cd},\ \   \ \   \ \  u_{xy,4}=\prod\limits_{z\in
        S(y)}\frac{u_{yz,3}b+u_{yz,4}cd+u_{yz,5}b}
    {u_{yz,1}b+u_{yz,2}a+cd},\\[4mm]
    u_{xy,5}=\prod\limits_{z\in
        S(y)}\frac{u_{yz,6}a+u_{yz,7}bd+u_{yz,8}c}
    {u_{yz,1}b+u_{yz,2}a+cd}, \ \     u_{xy,6}=\prod\limits_{z\in
        S(y)}\frac{c+u_{yz,1}b+u_{yz,2}ad}
    {u_{yz,1}b+u_{yz,2}a+cd}\\[4mm]
    u_{xy,7}=\prod\limits_{z\in
        S(y)}\frac{u_{yz,3}b+u_{yz,4}c+u_{yz,5}bd}
    {u_{yz,1}b+u_{yz,2}a+cd},\ \ 
    u_{xy,8}=\prod\limits_{z\in
        S(y)}\frac{u_{yz,6}a+u_{yz,7}b+u_{yz,8}cd}
    {u_{yz,1}b+u_{yz,2}a+cd},\\
    \end{array}\right.
    \end{equation}
    where, as before $a=\exp(\beta \ab)$, $b=\exp(\beta \bb)$, $c=\exp(\beta \cb)$ and $d=\exp(\beta J)$, and
    \begin{equation}\label{denuh}
    \begin{array}{ll}
    u_{xy,1}=\exp\left(-\frac{1}{2} h_{xy,\e_1\e_2}-h_{xy,\e_1\e_1}\right), \ \   u_{xy,2}=\exp\left(-\frac{1}{2} h_{xy,\e_1\e_3}-h_{xy,\e_1\e_1}\right),\\[2mm]
    u_{xy,3}=\exp\left(-\frac{1}{2} h_{xy,\e_2\e_1}-h_{xy,\e_1\e_1}\right),\ \     u_{xy,4}=\exp\left(h_{xy,\e_2\e_2}-h_{xy,\e_1\e_1}\right),\\[2mm]
    u_{xy,5}=\exp\left(-\frac{1}{2} h_{xy,\e_2\e_3}-h_{xy,\e_1\e_1}\right),\ \     u_{xy,6}=\exp\left(-\frac{1}{2} h_{xy,\e_3\e_1}-h_{xy,\e_1\e_1}\right),\\[2mm]
    u_{xy,7}=\exp\left(-\frac{1}{2} h_{xy,\e_3\e_2}-h_{xy,\e_1\e_1}\right),\ \      u_{xy,8}=\exp\left(h_{xy,\e_3\e_3}-h_{xy,\e_1\e_1}\right).\\[2mm]
    \end{array}
    \end{equation}
\end{cor}
It is worth mentioning that there are infinitely many solutions of
the system \eqref{necessary} corresponding to each solution of the
system of equations \eqref{canonic_u}. However, we show that each
solution of the system \eqref{canonic_u} uniquely determines a
Gibbs measure. We denote by $\mu_{\bf{u}}$ the Gibbs measure
corresponding to the solution  $\bf{u}$
of \eqref{canonic_u}.

\begin{thm}\label{ccu}
    There exists a unique Gibbs measure $\mu_\mathbf{u}$ associated with
    the function $\mathbf{u}=\{\mathbf{u}_{xy}, \ \langle{x,y}\rangle\in
    L \}$  where $\mathbf{u}_{xy}=(u_{xy,1},u_{xy,2},u_{xy,3},u_{xy,4},u_{xy,5},u_{xy,6},u_{xy,7},u_{xy,8})$ is a
    solution of the system \eqref{canonic_u}.
\end{thm}
\begin{proof}
    Let $\mathbf{u}=\{\mathbf{u}_{xy}, \ \langle{x,y}\rangle\in L \}$ be
    a  function, where $\mathbf{u}_{xy}=(u_{xy,1},u_{xy,2},u_{xy,3},u_{xy,4},u_{xy,5},u_{xy,6},\\u_{xy,7},u_{xy,8})$ is
    a solution of the system \eqref{canonic_u}. Then, for any
    $h_{xy,++}\in\br$ a function $\mathbf{h}=\{\mathbf{h}_{xy},\
    \langle{x,y}\rangle\in L\}$ defined by

    \begin{eqnarray*}
        \mathbf{h}_{xy}&=&\big\{h_{xy,\eta_1,\eta_1},\
        \log(u_{xy,1})+h_{xy,\eta_1,\eta_2}, \
        \log(u_{xy,2})+h_{xy,\eta_1,\eta_3}, \
        \log(u_{xy,3})+h_{xy,\eta_2,\eta_1}, \
        \log(u_{xy,4}) \\
        &&+h_{xy,\eta_2,\eta_2},\log(u_{xy,5})+h_{xy,\eta_2,\eta_3}, \
        \log(u_{xy,6})+h_{xy,\eta_3,\eta_1}, \
        \log(u_{xy,7})+h_{xy,\eta_3,\eta_2}, \
        \log(u_{xy,8}) \\
        &&+h_{xy,\eta_3,\eta_3} \big\}
    \end{eqnarray*}
    is a solution of \eqref{necessary}.

    Now fix $n\geq1$. Since $|W_{n-1}|=k^{n-1}$ and $|S(x)|=k$ we get
    $|L_{n}\setminus L_{n-1}|=k^n$. Let $\s$ be any configuration on
    $\Om_{V_n}$. Denote
    \[
    \begin{array}{ll}
    \cn_{1,n}(\s)=\{\langle{x,y}\rangle\in{L_n\setminus{L_{n-1}}}:\ \s(x)=\eta_1,\ \s(y)=\eta_1,\ x\in W_{n-1},\ y\in S(x)\}\\
    \cn_{2,n}(\s)=\{\langle{x,y}\rangle\in{L_n\setminus{L_{n-1}}}:\ \s(x)=\eta_1,\ \s(y)=\eta_2,\ x\in W_{n-1},\ y\in S(x)\}\\
    \cn_{3,n}(\s)=\{\langle{x,y}\rangle\in{L_n\setminus{L_{n-1}}}:\ \s(x)=\eta_1,\ \s(y)=\eta_3,\ x\in W_{n-1},\ y\in S(x)\}\\
    \cn_{4,n}(\s)=\{\langle{x,y}\rangle\in{L_n\setminus{L_{n-1}}}:\ \s(x)=\eta_2,\ \s(y)=\eta_1,\ x\in W_{n-1},\ y\in S(x)\}\\
    \cn_{5,n}(\s)=\{\langle{x,y}\rangle\in{L_n\setminus{L_{n-1}}}:\ \s(x)=\eta_2,\ \s(y)=\eta_2,\ x\in W_{n-1},\ y\in S(x)\}\\
    \cn_{6,n}(\s)=\{\langle{x,y}\rangle\in{L_n\setminus{L_{n-1}}}:\ \s(x)=\eta_2,\ \s(y)=\eta_3,\ x\in W_{n-1},\ y\in S(x)\}\\
    \cn_{7,n}(\s)=\{\langle{x,y}\rangle\in{L_n\setminus{L_{n-1}}}:\ \s(x)=\eta_3,\ \s(y)=\eta_1,\ x\in W_{n-1},\ y\in S(x)\}\\
    \cn_{8,n}(\s)=\{\langle{x,y}\rangle\in{L_n\setminus{L_{n-1}}}:\ \s(x)=\eta_3,\ \s(y)=\eta_2,\ x\in W_{n-1},\ y\in S(x)\}\\
    \cn_{9,n}(\s)=\{\langle{x,y}\rangle\in{L_n\setminus{L_{n-1}}}:\
    \s(x)=\eta_3,\ \s(y)=\eta_3,\ x\in W_{n-1},\ y\in S(x)\}
    \end{array}
    \]

    We have

    \begin{eqnarray*}
        \prod_{x\in W_{n-1}\atop{y\in
                S(x)}}\exp\left\{h_{xy,\s(x)\s(y)}\s(x)\s(y)\right\}&=&
        \prod\limits_{\langle
            x,y\rangle\in\cn_{1,n}(\s)}\exp\left\{h_{xy,\eta_1\eta_1}\right\}
        \prod\limits_{\langle
            x,y\rangle\in\cn_{2,n}(\s)}u_{xy,1}\cdot\exp\left\{h_{xy,\eta_1\eta_1}\right\}  \\
        &&\times\prod\limits_{\langle
            x,y\rangle\in\cn_{3,n}(\s)}u_{xy,2}\cdot\exp\left\{h_{xy,\eta_1\eta_1}\right\}
        \prod\limits_{\langle
            x,y\rangle\in\cn_{4,n}(\s)}u_{xy,3}\cdot\exp\left\{h_{xy,\eta_1\eta_1}\right\}  \\
        &&\times\prod\limits_{\langle
            x,y\rangle\in\cn_{5,n}(\s)}u_{xy,4}\cdot\exp\left\{h_{xy,\eta_1\eta_1}\right\}
        \prod\limits_{\langle
            x,y\rangle\in\cn_{6,n}(\s)}u_{xy,5}\cdot\exp\left\{h_{xy,\eta_1\eta_1}\right\}  \\
        &&\times\prod\limits_{\langle
            x,y\rangle\in\cn_{7,n}(\s)}u_{xy,6}\cdot\exp\left\{h_{xy,\eta_1\eta_1}\right\}
        \prod\limits_{\langle
            x,y\rangle\in\cn_{8,n}(\s)}u_{xy,7}\cdot\exp\left\{h_{xy,\eta_1\eta_1}\right\}  \\
        &&\times\prod\limits_{\langle
            x,y\rangle\in\cn_{9,n}(\s)}u_{xy,8}\cdot\exp\left\{h_{xy,\eta_1\eta_1}\right\}  \\
    \end{eqnarray*}
    \begin{eqnarray*}
        =\prod\limits_{\langle x,y\rangle\in L_n\setminus
            L_{n-1}}\exp\left\{h_{xy,\eta_1\eta_1}\right\}
        \prod\limits_{\langle x,y\rangle\in\cn_{2,n}(\s)}u_{xy,1}
        \prod\limits_{\langle x,y\rangle\in\cn_{3,n}(\s)}u_{xy,2}
        \prod\limits_{\langle x,y\rangle\in\cn_{4,n}(\s)}u_{xy,3}
        \prod\limits_{\langle x,y\rangle\in\cn_{5,n}(\s)}u_{xy,4} \\
        \times\prod\limits_{\langle x,y\rangle\in\cn_{6,n}(\s)}u_{xy,5}
        \prod\limits_{\langle x,y\rangle\in\cn_{7,n}(\s)}u_{xy,6}
        \prod\limits_{\langle x,y\rangle\in\cn_{8,n}(\s)}u_{xy,7}
        \prod\limits_{\langle x,y\rangle\in\cn_{9,n}(\s)}u_{xy,8}
    \end{eqnarray*}

    By means of the last equality, from \eqref{mu} and \eqref{Zn} we
    find
    $$
    \mu_{\h}^{(n)}(\s)=\frac{\exp\{-\b H_n(\s)\}\prod\limits_{x\in
            W_{n-1}\atop{y\in
                S(x)}}\exp\left\{h_{xy,\s(x)\s(y)}\s(x)\s(y)\right\}}
    {\sum\limits_{\w\in\Om_{V_n}}\exp\{-\b H_n(\w)\}\prod\limits_{x\in
            W_{n-1}\atop{y\in
                S(x)}}\exp\left\{h_{xy,\s(x)\w(y)}\s(x)\w(y)\right\}}
    $$

    One can see from \eqref{mu} and \eqref{Zn} does not depend to
    $h_{xy,\eta_1\eta_1}$. So, we can say that each solution $\mathbf{u}$ of the
    system \eqref{canonic_u} uniquely determines only one Gibbs measure
    $\mu_{\mathbf{u}}$.
\end{proof}

\begin{rem}
    Hence, due to Theorem \ref{ccu} a phase transition exists for
    the model if the equation
    \eqref{canonic_u} has at least two solutions.
\end{rem}

\section{The existence of  the phase transition}\label{e-Gibbs measures}

In this section, we are going to establish the existence of the phase transition, by analyzing the equation \eqref{canonic_u} for the model defined on the Cayley tree of order two, i.e. $k=2$.

We recall that ${\bf u}=\{{\bf u}_{xy}\}_{\langle{x,y}\rangle\in L}$
is a translation-invariant function, if one has
$\mathbf{u}_{xy}=\mathbf{u}_{zw}$ for all
$\langle{x,y}\rangle,\langle{z,w}\rangle\in L$. A measure $\m_{\bf
    u}$, corresponding to a translation-invariant function ${\bf u}$,
is called a {\it translation-invariant Gibbs measure}.

Solving the equation \eqref{canonic_u}, in general, is rather very
complex. Therefore, let us first restrict ourselves to the
description of its translation-invariant solutions. Hence,
\eqref{canonic_u} reduces to the following one

\begin{equation}\label{tru}
\left\{\begin{array}{ll}
u_1=\left(\frac{u_3bd+u_4c+u_5b}
{u_1b+u_2a+cd}\right)^2,  \ \   u_2=\left(\frac{u_6ad+u_7b+u_8c}{u_1b+u_2a+cd}\right)^2,\\[4mm]
u_3=\left(\frac{c+u_1bd+u_2a}
{u_1b+u_2a+cd}\right)^2, \ \ \ u_4=\left(\frac{u_3b+u_4cd+u_5b}{u_1b+u_2a+cd}\right)^2,\\[4mm]
u_5=\left(\frac{u_6a+u_7bd+u_8c}
{u_1b+u_2a+cd}\right)^2, \ \ u_6=\left(\frac{c+u_1b+u_2ad}{u_1b+u_2a+cd}\right)^2,\\[4mm]
u_7=\left(\frac{u_3b+u_4c+u_5bd}
{u_1b+u_2a+cd}\right)^2, \ \ u_8=\left(\frac{u_6a+u_7b+u_8cd}
{u_1b+u_2a+cd}\right)^2.\\
\end{array}\right.
\end{equation}

Now, let us assumethat $a=b$, and consider the following set:

\begin{equation}
A=\{(u_1,\cdots,u_8): u_1=u_2=u_3=u_5=u_6=u_7,\ u_4=u_8=1\}
\end{equation}
which is invariant w.r.t. \eqref{tru}.
Therefore, we consider \eqref{tru} over $A$, hence the reduced equation has the following form:

\begin{equation}\label{xyzre}
u=\left(\dfrac{u(a+ad)+c}{2ua+cd}\right)^2
\end{equation}

Denoting
\begin{equation}
\a=\frac{4c}{a(1+d)^3}, \ \Upsilon=\frac{d+d^2}{2}, \  X=\frac{u(a+ad)}{c},
\end{equation}
we rewrite \eqref{xyzre} as follows
\begin{eqnarray}\label{UU}
\a X=\left(\dfrac{1+X}{\Upsilon+X}\, \right)^2.
\end{eqnarray}

To solve the last equation, we apply the following well-known fact \cite[Proposition 10.7]{22} and adopt it to our setting.

\begin{lem}\label{lem1}
    \begin{itemize}
        \item[(1).] If $\Upsilon \leq 9$ then  \eqref{UU} has a unique solution.

        \item[(2).] If $\Upsilon > 9$ then there are $\zeta_1$ and $\zeta_2$ such that $0 < \zeta_1 < \zeta_2$, and if $\zeta_1 < \a< \zeta_2$ then \eqref{UU} has three solutions.

        \item[(3).] If $\a=\zeta_1$ and $\a=\zeta_2$ then \eqref{UU} has two solutions. 
     \end{itemize}     
       
        The quantities $\zeta_1$ and $\zeta_2$ are determined from the formula
        \begin{eqnarray}\label{eta}
        \zeta_i=\dfrac{1}{v_i}\left(\dfrac{1+v_i}{\Upsilon+v_i}\, \right)^2, \ i=1,2,
        \end{eqnarray}
where $v_1$ and $v_2$ are solutions to the equation $v^2+(3-\Upsilon)v+\Upsilon =0$.  
\end{lem}

Now the condition $\Upsilon > 9$ is reduced to 
$$
d^2+d-18>0
$$
which with the positivity of $d$ implies 
$$
d>\frac{\sqrt{73}-1}{2}.
$$
Hence, the last condition is a necessary condition for the existence of three solutions of \eqref{UU}. 

The condition $\zeta_1 < \a< \zeta_2$ ensures the existence of the translation-invariant solutions of \eqref{tru}, which implies the occurrence of the phase transition for the considered model.
Therefore, let us rewrite the last condition in terms of $\Upsilon$. 
One can calculate that
$$
v_{1,2}=\frac{1}{2}\big(\Upsilon-3\pm\sqrt{(\Upsilon-9)(\Upsilon-1)}\big).
$$
Then $\zeta_{1,2}$ has the following form
$$
\zeta_{1,2}=\frac{2(\Upsilon-5\pm\sqrt{(\Upsilon-9)(\Upsilon-1)})}{(\Upsilon-3\pm\sqrt{(\Upsilon-9)(\Upsilon-1)})(5\Upsilon-9\pm3\sqrt{(\Upsilon-9)(\Upsilon-1)})}.
$$

Hence, we can formulate the following result.

\begin{thm}\label{phase}
    If  $d>\frac{\sqrt{73}-1}{2}$ and
\begin{eqnarray*}
&&\a>\frac{2(\Upsilon-5-\sqrt{(\Upsilon-9)(\Upsilon-1)})}{(\Upsilon-3-\sqrt{(\Upsilon-9)(\Upsilon-1)})(5\Upsilon-9-3\sqrt{(\Upsilon-9)(\Upsilon-1)})}\\[2mm]
&&\a<\frac{2(\Upsilon-5+\sqrt{(\Upsilon-9)(\Upsilon-1)})}{(\Upsilon-3+\sqrt{(\Upsilon-9)(\Upsilon-1)})(5\Upsilon-9+3\sqrt{(\Upsilon-9)(\Upsilon-1)})}
\end{eqnarray*}
then there exists a phase transition for the $\l$-model with competing Potts interactions on the Cayley tree of order two.
\end{thm}

\smallskip


\medskip
\begin{thebibliography}{9}


\bibitem{2} R.J. Baxter, \textit{Exactly Solved Models in Statistical Mechanics}, (New York: Academic, 1982).

\bibitem{6} G. I. Botirov, U. A. Rozikov, \textit{Theo. Math. Phys.} {\bf 153}, p. 1423 (2007).

\bibitem{DGM}  S.N. Dorogovtsev, A.V. Goltsev, J.F.F. Mendes, \textit{Eur. Phys. J. B} {\bf 38} (2004) 177.


\bibitem{11} N. N. Ganikhodzhaev, {\it Theor. Math. Phys.}, {\bf 85}, 1125-1134 (1990).


\bibitem{5}  N. Ganikhodjaev, F.  Mukhamedov, J.F.F.  Mendes, {\it Jour. Stat. Mech}. 2006, P08012.

\bibitem{GMP} N. Ganikhodjaev, F.  Mukhamedov, C.H. Pah, {\it Phys. Lett. A. } {\bf 373}, 33--38 (2008).

\bibitem{12} N. N. Ganikhodjaev, U. A. Rozikov, {\it Osaka J. Math.},{\bf 37}, 373-383 (2000).

\bibitem{1} H.O. Georgii, \textit{Gibbs Measures and Phase Transitions}(de Gruyter Studies in Mathematics vol 9) (Berlin:de Gruyter, 1988)

\bibitem{4} R.A. Minlos, \textit{Introduction to Mathematical Statistical Physics}(Amer. Math. Soc., Providence, RI, 2000).


\bibitem{15} F. Mukhamedov, {\it Rep. Math. Phys.}{\bf 53}, p. 1-18 (2004).

\bibitem{MAK}   F.  Mukhamedov, H. Akin, O. Khakimov, {\it Jour. Stat. Mech}. 2017, P053208.

\bibitem{16} F. Mukhamedov, Ch.-H. Pah, H. Jamil. {\it J.  Phys.: Conf. Ser.} { \bf 819}, (2017) 012020

\bibitem{17} F. Mukhamedov, Ch.-H. Pah, M. Rahmatullaev, H. Jamil. {\it J.  Phys.: Conf. Ser.}{\bf 949}, (2017) 012021

\bibitem{18} F. Mukhamedov, C.-H. Pah,  H. Jamil, {\it Theor. Math. Phys.} {\bf 194}, 260-273 (2018)

\bibitem{7} F. Mukhamedov, U. Rozikov, J.F.F. Mendes, {\it Jour. Math. Phys.}
{\bf 48}, 013301 (2007).



\bibitem{20} O. Melnikov, R. I. Tyshkevich, V. A. Yemelichev, V. I. Sarvanov, \textit{Lectures on Graph Theory} (B. I.
Wissenschaftsverlag, Mannheim ,1994).

\bibitem{NS}   M.P. Nightingale, M. Schick, \textit{J. Phys. A: Math. Gen.} {\bf 15} (1982) L39.

\bibitem{21} M. M. Rahmatullaev, M.A. Rasulova, \textit{Siberian Adv. Math.},  {\bf 26}, p. 215-229 (2016)

\bibitem{Ras} M.A. Rasulava, \textit{Theor. Math. Phys.} {\bf 199} 586-592 (2019).

\bibitem{8} U. A. Rozikov, {\it Gibbs Measures on Cayley Trees,} (World Scientific, Hackensack, 2013).

\bibitem{10} U. A. Rozikov, M. M. Rakhmatullaev, {\it Theo. Math. Phys.} {\bf 160}, 1292 (2009).


\bibitem{13} U.A. Rozikov, R.M. Khakimov, \textit{Theor. Math. Phys.}
{\bf 175},  699-709 (2013).

\bibitem{14} U. A. Rozikov, \textit{Siberan Math.
    Jour.} {\bf 39},  373-380, (1998).

\bibitem{9} R. B. Potts, {\it Proc. Cambridge Philos. Soc.}, {\bf48}, 106-109 (1952).

\bibitem{22} C. J. Preston, \textit{Gibbs States on Countable Sets} (Cambridge Univ. Press, London,1974)

\bibitem{3} Y.G. Sinai, \textit{Theory of Phase Transitions: Rigorous Results}, (Pergamon Press, Oxford, 1982).


\bibitem{V} J. Vannimenus ,  \textit{Z. Phys. B} {\bf 43}(1981) 141.

\bibitem{10} F. Y. Wu, {\it Rev. Modern Phys.} {\bf54}, 235-268 (1982).



\end{thebibliography}
\end{document}